\documentclass[11pt]{amsart}
\usepackage{amsmath,stmaryrd,amsthm,tikz,calc}
\usepackage{color}
\usepackage[margin=1.25in]{geometry}
\usepackage{color}
\usepackage{moreverb}
\usepackage{graphicx} 
\usepackage{float}
\usepackage{pdfsync}
\usepackage{hyperref}  
\hypersetup{colorlinks=true}    
\usepackage{bbm}
\usepackage{tcolorbox}

\date{\today}

\newcommand{\Rl}{\mathbb{R}}
\newcommand{\E}{\mathbb{E}}

\newcommand{\Nl}{\mathbb{N}}
\newcommand{\Cx}{\mathbb{C}}

\newcommand{\Ir}{\mathbb{Z}}

\renewcommand{\H}{\mathcal{H}}
\newcommand{\A}{\mathcal{A}}
\newcommand{\cB}{\mathcal{B}}

\renewcommand{\ker}{\mathop{\rm ker}}
\newcommand{\diam}{\mathop{\rm diam}}

\newcommand{\Tr}{{\rm Tr}}

\newcommand{\be}{\begin{equation}}
\newcommand{\ee}{\end{equation}}
\newcommand{\bea}{\begin{eqnarray}}
\newcommand{\eea}{\end{eqnarray}}
\newcommand{\beann}{\begin{eqnarray*}}
\newcommand{\eeann}{\end{eqnarray*}}
\newcommand{\eq}[1]{(\ref{#1})}

\newtheorem{theorem}{Theorem}[section]
\newtheorem{proposition}[theorem]{Proposition}

\newtheorem{lemma}[theorem]{Lemma}
\newtheorem{definition}[theorem]{Definition}
\newtheorem{Remark}[theorem]{Remark}

 \numberwithin{equation}{section}

\renewcommand{\epsilon}{\varepsilon}

\def\varnothing{\emptyset}
\newcommand\supp{\text{supp}}
\def\Pr{\mathbb{P}}
\def\E{\mathbb{E}}
\begin{document}

\title{Slow propagation in some disordered quantum spin chains}

\author[B. Nachtergaele]{Bruno Nachtergaele}
\address{Department of Mathematics and Center for Quantum Mathematics and Physics\\
University of California, Davis\\
Davis, CA 95616, USA}
\email{bxn@math.ucdavis.edu}
\author[J. Reschke]{Jake Reschke}
\address{Department of Mathematics\\
University of California, Davis\\
Davis, CA 95616, USA}
\email{jreschke@math.ucdavis.edu}

\begin{abstract}
We introduce the notion of {\em transmission time} to study the dynamics of disordered quantum spin chains and prove results relating its 
behavior to many-body localization properties. We also study two versions of the so-called Local Integrals of Motion (LIOM) representation 
of spin chain Hamiltonians and their relation to dynamical many-body localization. We prove that uniform-in-time dynamical localization 
expressed by a zero-velocity Lieb-Robinson bound implies the existence of a LIOM representation of the dynamics as well as a weak 
converse of this statement. We also prove that for a class of spin chains satisfying a form of exponential dynamical localization, sparse 
perturbations result in a dynamics in which transmission times diverge at least as a power law of distance, with a power for which we provide 
lower bound that diverges with increasing sparseness of the perturbation.
\end{abstract}

\maketitle

\section{Introduction} 

Anderson localization in random Schr\"odinger operators is quite well understood. Mathematical proofs of this phenomenon have been given
under a variety of conditions. See the recent book by Aizenman and Warzel for an overview of the state-of-the-art \cite{aizenman:2015}.
The physical phenomenon is a drastic slowdown of transport in the system's dynamics, which is seen as the consequence of a 
change in the nature of the spectrum from continuous spectrum (extended states) to pure point spectrum (localized states). 

The problem of Many-Body Localization (MBL) is the question of what happens to localization properties in the presence of interactions. 
Although Anderson in his work that started the subject of localization \cite{anderson:1958} envisioned the phenomenon for interacting systems, 
research on MBL picked up only relatively recently stimulated by papers by Basko, Aleiner, and Altshuler \cite{basko:2006}, Oganesyan and Huse 
\cite{oganesyan:2007}, and Pal and Huse \cite{pal:2010}. 

Quantum spin system with, for example, nearest neighbor interactions, are among the simplest interacting quantum many-body systems 
and much of the recent work on MBL dealt with one of just three one-dimensional quantum spin models: the $XY$ chain, the quantum 
Ising chain, and the $XXZ$ chain. The small number of rigorous results that have been obtained so far are also mostly restricted to these three models. 
Exponential dynamical localization, uniformly in time, was proved for a class of disordered $XY$ chains by exploiting their connection to Anderson 
models \cite{hamza:2012,sims:2016,abdul-rahman:2017}. Imbrie studied the quantum Ising chain with random couplings and fields \cite{imbrie:2016}. 
Localization properties in the low-energy region, called the droplet-regime, of the ferromagnetic XXZ chain were proved in
\cite{elgart:2018b,elgart:2018,elgart:2018a,beaud:2017,beaud:2018}.

For a single quantum particle, the study of localization for a long time focused on spectral properties. i.e., proving the occurrence of point spectrum 
with associated eigenvectors that satisfy exponential decay. Later, multi-scale analysis \cite{germinet:2001} and the fractional-moment method \cite{aizenman:1993} emerged as two powerful tools to study dynamical localization. Systems of $N$ interacting particles can be analyzed by extending these methods, as along as $N$ is fixed \cite{chulaevsky:2009,aizenman:2009}.

The first main result of this work is the proof of a relation between uniform dynamical localization and the existence of Local Integrals of Motion (LIOM). The LIOM picture \cite{serbyn:2013a,chandran:2015} has been proposed as the mechanism by which systems exhibiting MBL do not thermalize under their own (closed system) dynamics and, in particular, that violate the Eigenfunction Thermalization Hypothesis (ETH). We give two definitions of 
LIOMs, consistent with the different ways this concept has been considered in the literature. For lack of a better name, we call them LIOMs of the first kind {Definition \ref{def:LIOM}} and LIOMS of the second kind (Definition \ref{def:LIOM2}). The first kind implies dynamical localization of the form generically expected for strongly disordered quantum spin chains. The second kind, as we show, exist when we have uniform-in-time dynamical localization, such as has been proved to occur in the random $XY$ chain \cite{hamza:2012}.

In interacting many-body systems it is most natural to express localization in terms of dynamical properties directly. A good (but not typcial) example is the zero-velocity Lieb-Robinson bound proved for the disordered $XY$ chain in \cite{hamza:2012}. In this work, we introduce the notion of {\em transmission time}, as the smallest time a signal or disturbance can reach a prescribed strength a given distance away from the source.
See Definition \ref{def:transmission_time}. For exponentially localized systems, we expect transmission times grow exponentially with the distance.
We then prove that exponentially localized systems perturbed by sparse disorder, have transmission times that grow at least as a 
power law and we we give a lower bound for the power that diverges with increasing sparseness of the perturbation. A large power 
indicates sub-diffusive behavior. We model the sparse disorder by adding a uniformly bounded but otherwise arbitrary nearest-neighbor term to the Hamiltoian at locations determined by a Bernoulli process with small probability of success.

De Roeck and coworkers have argued that MBL, interpreted as the complete absence of transport, is only possible in one-dimensional systems. 
They argue that diffusion of energy is inevitable in higher dimensions 
\cite{de-roeck:2015,de-roeck:2016,de-roeck:2017,de-roeck:2017a,luitz:2017,thiery:2018}. We only study one 
dimensional systems in this work, and therefore we do not have results that either support or contradict these arguments.  Rather, for 
one-dimensional systems our results implies a degree of robustness of localization phenomena in the sense of slow propagation. Others have investigated stability of MBL in spin chains under the influence of regions of low disorder or coupling to a heat bath \cite{goihl:2019}, in a kicked quantum spin chain model \cite{braun:2019} and by extensive numerical calculation for the Heisenberg chain \cite{suntajs:2019}. The latter studies consider properties of the spectral form factor (i.e., the Fourier transform of a two-point function) to look for an indicator of an MBL-type transition. It would be interesting to 
supplement these studies with information about transmission times in these models.

In Section \ref{sec:mainresults} we introduce several definitions related to MBL and describe our main results. The proofs are in 
Section \ref{sec:proofs}. Two applications are discussed in Section \ref{sec:applications}. Some auxiliary facts are collected in an appendix.

%
%
%
%
%
\section{Many-body localization properties and main results}\label{sec:mainresults}

In this section we define several properties associated with localized many-body systems. We focus on characteristics of the dynamics in terms of which our main results are formulated and restrict ourselves to the one-dimensional setting. All notions make sense for multi-dimensional systems but, as discussed in the introduction, the phenomenon of many-body localization as it is commonly understood may well be restricted to one dimension.

We will consider subsystems of a chain of quantum systems labeled by $x\in\Ir$, with a finite-dimensional Hilbert space $\H_x$ for each $x\in\Ir$. 
The Hilbert space of the subsystem associated with a finite set $X\subset\Ir$, is given by $\H_X=\bigotimes_{x\in X} \H_x$, and the observables 
measurable in this subsystem are given by $\A_X:=\cB(\H_X)$. The elements of $\A^{\text loc}:= \bigcup_{X\subset\Ir} \A_X$, where the union is over finite subsets, are called the {\em local observables}, whereas the norm completion of $\A^{\text loc}$, denoted by $\A_\Ir$, is the algebra of 
quasi-local observables. We denote the closed unit ball of $\A_X$ by $\A_X^1$.

A convenient way to specify a model is with an interaction, which is a map $\Phi$ assigning to each finite set $X\subset\Ir$ an element $\Phi(X)=\Phi(X)^*\in \A_X$. Associated to the interaction $\Phi$ is the family of local Hamiltonians $H_\Lambda=\sum_{X\subset\Lambda}\Phi(X)\in\A_\Lambda$, defined for each finite subset $\Lambda\subset\Ir$. The Heisenberg dynamics generated by a family of local Hamiltonians determined by an interaction $\Phi$ is defined in the usual way:
\begin{equation}
\tau_t^{H_\Lambda}(A)=e^{itH_\Lambda}Ae^{-itH_\Lambda}
\end{equation}

The interactions $\Phi$ may be random, meaning the following: There is a probability space $(\Omega,\mathcal{F},\Pr)$, and to each $\omega\in\Omega$ there is assigned an interaction $\Phi(\omega)$. We assume weak measurability of the random operators $\omega\mapsto \Phi(\omega)(X)$ for each finite $X\subset\Ir$.

A finite range interaction is one for which there exists $R\geq 0$ such that $\Phi(X)=0$ unless $\diam X \leq R$. $R$ is then the {\em range} of the 
interaction. A common way to introduce a model with a finite-range interaction is to specify self-adjoint $h_x\in \A_{[x,x+R]}$, for each $x\in \Ir$.

\subsection{Dynamical Localization}

In the single-particle setting, dynamical localization refers to the absence of ballistic  or diffusive propagation in the system's Schr\"{o}dinger evolution. Initially localized wave functions remain localized for all time under the dynamics. A natural analogue of this property in the setting of quantum spin chains is localization of the Heisenberg dynamics. We consider a general notion of dynamical localization expressed by the following definition.

\begin{definition}\label{def:dynlocaliz} Let $F:\Ir_+\to (0,\infty)$ be a non-increasing function with the property $\lim_{x\to\infty}F(x)=0$.\newline
(i) We say that a family $\{H_\Lambda:\Lambda\subset\Ir \text{ finite intervals}\}$ of random local Hamiltonians $H_\Lambda\in\A_\Lambda$ exhibits dynamical localization with decay function $F$
 if there exists a constant $\beta\geq0$ and a function $\chi:\Nl\to(0,\infty)$ such that for any sets $X,Y\subseteq \Lambda$ with $Y\subset[\min X,\max X]^c$, the random variable
\be\label{heislocal}
C_{\Lambda;X,Y}\equiv\sup_{t\in\Rl}\sup_{\substack{A\in\mathcal{A}_X^1\\ B\in\mathcal{A}_Y^1}}\frac{\Vert [\tau_t^{H_\Lambda}(A),B] \Vert}{\chi(|X|)(1+|t|^\beta)}
\ee
satisfies
\begin{equation}\label{uniform_exp_loc}
\E C_{\Lambda;X,Y}\leq F(d(X,Y))
\end{equation}
Here $d(X,Y)=\min\{|x-y|: x\in X,y\in Y\}$ is the usual set distance. 

\noindent
(ii) If $F$ is of the form $F(x)=e^{-\eta x}$ we say the family $\{H_\Lambda\}$ exhibits exponential dynamical localization. In this case $\eta^{-1}$ is called (a bound for the) {localization length}.

\noindent
(iii) If $F$ is of the form $F(x)=e^{-\eta x^\rho}$ for some $\rho\in (0,1)$, we say the family $\{H_\Lambda\}$ exhibits stretched exponential dynamical localization.

\noindent(iv) We say the family $\{H_\Lambda\}$ exhibits dynamical localization with decay function $F$ uniformly in time if it satisfies (i) with $\beta=0$.
\end{definition}

The following lemma shows that if a family of local Hamiltonians is dynamically localized and the corresponding family of local dynamics has a thermodynamic limit, then the infinite volume dynamics is also dynamically localized with the same decay function.

\begin{lemma}
Suppose that $\{H_\Lambda\}$ is a family of dynamically localized Hamiltonians with decay function $F$, and that the corresponding family of dynamics $\{\tau_t^{H_\Lambda}\}$ has a thermodynamic limit. In other words, there is an exhaustive sequence $\Lambda_n\uparrow\Ir$ such that almost surely,
\be
\lim_{n\to\infty}\tau_t^{H_{\Lambda_n}}\equiv \tau_t
\ee
strongly for all $t\in\Rl$, where $\tau_t$ is a $*$-automorphism of $\A_\Ir^\text{loc}$. Then for any finite set $X\subset \Ir$ and any set $Y\subseteq [\min X,\max X]^c$, the random variable
\be
C_{X,Y}\equiv\sup_{t\in\Rl}\sup_{\substack{A\in\mathcal{A}_X^1\\ B\in\mathcal{A}_Y^1}}\frac{\Vert [\tau_t(A),B] \Vert}{\chi(|X|)(1+|t|^\beta)}
\ee
satisfies
\be
\E C_{X,Y}\leq F(d(X,Y))
\ee
\end{lemma}

\begin{proof}
First let $X,Y\subset\Ir$ be finite, with $Y\subset [\min X,\max X]^c$. It follows immediately that,
\be
C_{ X,Y}=\sup_{t\in\Rl}\sup_{\substack{A\in\mathcal{A}_X^1\\ B\in\mathcal{A}_Y^1}}\frac{\Vert [\tau_t(A),B] \Vert}{\chi(|X|)(1+|t|^\beta)}\leq \liminf_{n\to\infty} C_{\Lambda_n;X,Y}.
\ee
By Fatou's lemma, $\E C_{ X,Y} \leq F(d(X,Y))$. Now suppose $Y\subseteq[\min X,\max X]^c$ is infinite. For any sequence of finite sets $Y_n\uparrow Y$, by using local approximations and the fact that $C_{X,Y_n}$ is monotone in $n$ we obtain
\be
C_{X,Y}\leq \lim_{n\to\infty} C_{X,Y_n},
\ee
which proves the lemma.
\end{proof}

\subsection{Local Integrals of Motion} 

The lack of ergodicity seen in MBL systems can be `explained'  as a consequence the emergence of an extensive set of local conserved quantities, called local integrals of motion (LIOMs). In this section we propose precise definitions of LIOMs. Heuristic definitions of LIOMs have been given in the physics literature, \cite{huse:2014}, \cite{serbyn:2013}. LIOMs are thought to account for most of the phenomena of MBL. See, for example, the review paper \cite{imbrie:2017}. To address the variety seen in the physics literature we formulate two distinct definitions. Specifically, Definition \ref{def:LIOM} given below is modeled after the discussion in \cite{huse:2014}, while Definition \ref{def:LIOM2} was motivated by \cite{chandran:2015}. We refer to them as LIOMs of the first kind and LIOMs of the second kind, respectively. We briefly discuss the relation between the two at the end of this section.

 In the following definition we restrict our attention to quantum spin chains, for simplicity. The definition can also be formulated in higher-dimensions. Let $d_x\geq 2$ denote the dimension of the Hilbert space at $x\in\Ir$.

\begin{definition}[{\bf LIOMs of the first kind}]\label{def:LIOM}
Let $H_n\in\mathcal{A}_{[0,n]}$ be a sequence of random Hamiltonians. We say that the sequence $H_n$ has LIOMs of the first kind if the following conditions are satisfied:
\begin{enumerate}
\item There is a sequence of random unitary maps $U_n\in\mathcal{A}_{[0,n]}$ such that
\begin{equation}
U_n^*H_nU_n=\sum_{X\subseteq[0,n]}\sum_{{\bf m}\in\prod_{x\in X} \{2,...,d_x\}}\phi_n({\bf m},X)\prod_{x\in X}S_{{\bf m}_x;x},
\end{equation}
where $S_{m;x}$ is the operator supported at the site $x$ given by the matrix,
\be
(S_{m;x})_{jk}=\delta_{j,1}\delta_{k,1}-\delta_{j,m}\delta_{k,m}
\ee
 and the $\phi_n({\bf m},X)$ are random variables satisfying
\begin{equation}\label{liom2pt}
\sup_{n}\, \E \left[\sup_{x,y\in[0,n]} \frac{1}{F(|x-y|)} \sum_{\substack{X\subseteq[0,n]:\\ x,y\in X}}\left\|\sum_{{\bf m}\in\prod_{x\in X} \{2,...,d_x\}}\phi_n({\bf m},X)\prod_{x\in X} S_{{\bf m}_x;x}\right\| \right]<\infty.
\end{equation}
for some non-increasing function $F:\Ir_+\to(0,\infty)$ satisfying $\lim_{x\to\infty}F(x)=0$.
\item The sequence of unitary maps $U_n$ is quasi-local, in the sense that for all disjoint finite subsets $X,Y\subset\Gamma$,
\begin{equation}\label{liomqloc}
\sup_n\E \sup_{\substack{A\in\mathcal{A}_X^1\\ B\in\mathcal{A}_Y}}\|[ U_n^*AU_n,B]\| \leq \sum_{\substack{x\in X\\y\in Y}}G(|x-y|),
\end{equation}
for some non-increasing function $G:\Ir_+\to(0,\infty)$ satisfying $\lim_{x\to\infty}G(x)=0$.
\end{enumerate}
\end{definition}

\begin{Remark}
The LIOMs in definition \ref{def:LIOM} are the quasi-local operators $U_nS_{m;x}U_n^*$. The key feature of the family $\{S_{m;x}\}_{m=2}^{d_x}$ is that the operators are uniformly bounded, are mutually commuting, and generate a maximal abelian subalgebra of observables. Any other set of observables with these properties could be used in the definition instead.
\end{Remark}

The following theorem shows that the Heisenberg dynamics generated by a Hamiltonian with LIOMs of the first kind satisfies the type of 
propagation bound expressing dynamical localization.

\begin{theorem}\label{prop:LIOM_LRbound}
Suppose that the sequence of Hamiltonians $H_n$ has LIOMs of the first kind. Let $X$ and  $Y$ be finite disjoint subsets of $\Ir_+$. For a set $Z\subset\Ir_+$, let $Z_{n,\lambda}=\{x\in[0,n]: d(x,Z)\leq \lambda d(X,Y)\}$. Then for $\lambda\in(0,1/2)$,
\begin{equation}
\sup_{\substack{A\in\mathcal{A}_X^1\\ B\in\mathcal{A}_Y^1}}\|[\tau_t^{H_n}(A),B]\|\leq 2\bigg[ D_{n,X,\lambda}+D_{n,Y,\lambda}+|t|C_n \sum_{\substack{x\in X_{n,\lambda}\\y\in Y_{n,\lambda}}} F(|x-y|)\bigg],
\end{equation}
where $D_{n,X,\lambda}$ and $D_{n,Y,\lambda}$ are nonnegative random variables satisfying,
\begin{align}
\E D_{n,X,\lambda}\leq \sum_{\substack{x\in X\\y\in X_{n,\lambda}^c}}G(|x-y|) \text{\,\,\, and\,\,\, }\E D_{n,Y,\lambda}\leq \sum_{\substack{x\in Y\\ y\in Y_{n,\lambda}^c}}G( |x-y|),
\end{align}
and 
\begin{equation}
C_n(\omega)=\sup_{x,y\in[0,n]} \frac{1}{F(|x-y|)} \sum_{\substack{X\subseteq[0,n]:\\ x,y\in X}}\left\|\sum_{{\bf m}\in\prod_{x\in X} \{2,...,d_x\}}\phi_n({\bf m},X)\prod_{x\in X} S_{{\bf m}_x;x}\right\|,
\end{equation}
where by the assumptions in Definition \ref{def:LIOM} we have $\sup_n \E C_n<\infty$.
\end{theorem}

The proof of this theorem is given in Section \ref{sec:proofs_LIOM}.

It is natural to ask whether the existence of LIOMs also follows from dynamical localization. Indeed, the existence of LIOMs and dynamical localization 
are regarded as equivalent properties in the physics literature. It turns out to be convenient to use a slightly different notion of LIOMs to prove a result in this direction.

\begin{definition}[{\bf LIOMs of the second kind}]\label{def:LIOM2}
Suppose that $\Phi$ is a (random) finite range interaction with a thermodynamic limit $\tau_t$ generated by the derivation $\delta$. We say the interaction has LIOMs of the second kind if there exists a family $\{I_x\}_{x\in\Ir}$ of self-adjoint, uniformly bounded quasi-local observables $I_x$ satisfying the following:
\begin{enumerate}
\item There is a non-increasing function $F:\Ir_+\to (0,\infty)$, with $\lim_{n\to\infty} F(n) =0$, such that for all $x\in\Ir$,
\be
\E \sup_{A\in\mathcal{A}^1_Y}\|[ I_x,A]\| \leq F(d(x,Y)).
\ee
\item For each $x\in\Ir$,
\be
\delta(I_x)=0.
\ee
\item For each $A\in\mathcal{A}^\text{loc}$,
\be
\delta(A)=\lim_{n\to\infty} \sum_{x=-n}^n [I_x,A],
\ee
almost surely, i.e. the family $\sum_{x=-n}^nI_x$ of quasi-local Hamiltonians almost surely generate the same dynamics in the thermodynamic limit as $\Phi$.
\end{enumerate}
\end{definition}

\begin{Remark}
In Definition \ref{def:LIOM2} we do not assume that the LIOMs $I_x$ commute. From the time invariance it is necessary that $I_x\in\ker\delta$, thus if $\ker\delta$ is abelian the LIOMs will commute. We expect $\ker\delta$ to be abelian almost surely, generically for continuous randomness. Note that in finite volumes, $\delta(\cdot)=[H,\cdot]$ for a local Hamiltonian $H$, and simplicity of the spectrum of $H$ is equivalent to $\ker\delta$ being an abelian algebra.
\end{Remark}

The following proposition connects dynamical localization uniform in time with the `canonical LIOMs' introduced in \cite{chandran:2015}.

\begin{theorem}\label{thm:thermolioms}
Suppose a model with finite-range interactions is dynamically localized with decay function $F$ uniformly in time ($\beta=0$), and that $F$ has a finite first moment: $\sum_{x=1}^\infty xF(x)<\infty$. Then the model has LIOMs of the second kind. Moreover, a LIOM representation
(canonical in the sense of  \cite{chandran:2015}) can be given explicitly by the following expression:
\be
\tilde h_x = \lim_{n\to\infty} \frac{1}{T_n}\int_0^{T_n} \tau_t (h_x) dt.
\label{time_average}\ee
where $T_n$ is a suitably chosen (random) strictly increasing sequence in $\Nl$. The terms $\tilde h_x$ are time-invariant, and there is a constant $C>0$ such that
\be
\E ( \sup_{B\in\A_Y^1}\Vert [ \tilde h_x, B]\Vert ) \leq CF(d(x,Y))
\ee
for every $x\in \Ir$.
\end{theorem}
The proof of this result can be found in Section \ref{sec:proofs_LIOM}.

In the definition of LIOMs of the first kind, Definition \ref{def:LIOM}, nothing is said on the dependence of the unitaries and the interaction coefficients on the length, $n$, of the chain. One could expect however, that a random interaction $\Phi$ can be defined by 
\begin{equation}
\Phi(X) = \lim_{n\to\infty} \sum_{{\bf m}\in\{1,...,d-1\}^{|X|}}\phi_n({\bf m},X)\prod_{x\in X} S_{{\bf m}_x;x},
\end{equation}
where it should be understood that $n$ here refers to a finite spin chain labeled by $[-n,n]$.
Using the notion of {\em local convergence in F-norm} (see \cite[Definition 3.7]{nachtergaele:2019}), it is then straightforward to define conditions that ensure the existence of a commuting family of LIOMs of the second kind.

\subsection{Transmission Times}

\begin{definition}\label{def:transmission_time}
Given a Hamiltonian $H\in\mathcal{A}_{[0,n]}$ and an $\epsilon>0$ define the transmission time, $t(\epsilon)$ of $H$ as,
\begin{equation}
t(\epsilon)=\inf\{|t|: \sup_{\substack{A\in\mathcal{A}_0^1\\ B\in\mathcal{A}_n^1}}\|[\tau_t^H(A),B]\|>\epsilon\}.
\end{equation}
\end{definition}

Suppose we have a sequence $H_n\in\mathcal{A}_{[0,n]}$ of Hamiltonians with associated transmission times $t_n(\epsilon)$. It is reasonable to expect that dispersive effects may cause the commutator defining the transmission time to never exceed some fixed $\epsilon>0$ for large values of $n$. If this occurs then $t_n(\epsilon)$ will cease to be a meaningful quantity. For this reason we should consider a sequence $\epsilon_n$, suitably decaying in $n$, and instead consider the sequence of transmission times $t_n(\epsilon_n)$. We note that some authors prefer the term `scrambling time' instead of transmission time \cite{chen:2019}.

A natural question to ask is whether the transmission time is consistent with the propagation bounds imposed by a Lieb-Robinson bound. Suppose that the sequence $H_n$ satisfies,
\begin{equation}
\sup_{\substack{A\in\mathcal{A}_x^1\\ B\in\mathcal{A}_y^1}}\|[\tau_t^{H_n}(A),B]\|\leq C(e^{\mu v|t|}-1)e^{-\mu|x-y|}
\end{equation}
for $x\neq y$, uniformly in $n$. Such bounds are known to hold for a broad class of quantum spin models on general lattices \cite{nachtergaele:2006}. The bound implies that,
\begin{equation}
t_n(\epsilon_n)\geq \frac{1}{\mu v} \log(1+\frac{\epsilon_ne^{\mu n}}{C}),
\end{equation}
in which case
\begin{equation}
\limsup_{n\to\infty}\frac{n}{t_n(\epsilon_n)}\leq v
\end{equation}
provided $\epsilon_n$ decays subexponentially in $n$. 

We consider slow transport in a quantum spin chain to be characterized by super-linear growth of the transmission time. For stretched exponential dynamically localized spin chains the transmission time grows as a stretched exponential, as the next proposition shows.

\begin{proposition}
Suppose that a sequence $H_n\in\A_{[0,n]}$ of random Hamiltonians exhibits dynamical localization with decay function $F$ given by $F(x)=e^{-\eta x^\rho}$ for some $\rho\in(0,1]$. Then for any positive $\gamma$ and $\alpha$ such that $\beta\gamma+\alpha<1$,
\be
\frac{e^{\gamma \eta n^\rho}}{t_n(e^{-\alpha\eta n^\rho})}\to0
\ee
almost surely. 
\end{proposition}
\begin{proof}
 For $\beta=0$ it is easy to see that $\Pr ( t_n(e^{-\alpha \eta n^\rho} )=\infty \text{ eventually} )=1$. Assume $\beta>0$. By assumption,
\be\label{eq:t_time_prop}
\sup_{\substack{A \in \A_0^1\\ B \in \A_n}}\Vert \tau^{H_n}_t(A),B]\Vert\leq \chi(1)C_n(1+|t|^\beta),
\ee
where $\E C_n \leq e^{-\eta n^\rho}$. Choose any $\delta$ such that $\beta \gamma+\alpha<\delta<1$. Let 
$$
A_n=\left\{ \chi(1)C_n\leq e^{-\delta \eta n^\rho} \right\}.
$$
By Markov's inequality,
$$
\Pr(A_n^c)\leq \chi(1)\frac{\E C_n}{e^{-\delta\eta n^\rho}}\leq \chi(1)e^{-(1-\delta)\eta n^\rho}.
$$
It follows from the Borel-Cantelli lemma that $\Pr(\mathbbm{1}_{A_n}=1 \text{ eventually})=1$. \eqref{eq:t_time_prop} implies that,
$$
\mathbbm{1}_{A_n}t_n(e^{-\alpha \eta n^\rho})^\beta\geq \mathbbm{1}_{A_n}\left(\frac{e^{-\alpha \eta n^\rho}}{\chi(1) C_n}-1\right)\geq (e^{(\delta-\alpha) \eta n^\rho}-1)\mathbbm{1}_{A_n}
$$
Therefore
$$
\mathbbm{1}_{A_n} \frac{e^{\gamma \eta n^\rho}}{t_n(e^{-\alpha \eta n^\rho})}\leq \frac{e^{\gamma \eta n^\rho}}{(e^{(\delta-\alpha)\eta n^\rho}-1)^{1/\beta}}
$$
Since $\gamma<(\delta-\alpha)/\beta$ and $\mathbbm{1}_{A_n}=1$ eventually with probability 1, it follows that $\frac{ e^{\gamma\eta n^\rho} }{ t_n(e^{-\alpha \eta n^\rho}) }\to0$ almost surely.\newline
\end{proof}

In the case of exponential dynamical localization there exists a family of perturbations such that the perturbed model still has long transmission times. These perturbations consist of additional nearest neighbor interactions that occur with low density at random positions. For this class of perturbations we can prove that the transmission time grows super linearly provided the perturbations are sufficiently sparse\footnote{
After this work appeared on the arXiv, similar perturbations were considered by De Roeck, Huveneers, and Olla, who proved subdiffusive dynamics in classical Hamiltonian chains \cite{de-roeck:2020}.}. 

\begin{theorem}\label{thm:mainthm}
Let $H_n^0\in\mathcal{A}_{[0,n]}$ be a sequence of random Hamiltonians defined over the probability space $(\Omega_0,\Pr_0)$ which are exponentially dynamically localized in the sense of Definition \ref{def:dynlocaliz} ($\rho=1$). Let $(\delta_x)_{x=0}^\infty$ be an i.i.d. sequence of Bernoulli random variables over the probability space $(\Omega_1,\Pr_1)$, with $\Pr_1(\delta_0=0)=p\in(0,1]$. Let $(\psi_{x})_{x=0}^\infty$ denote a uniformly bounded sequence with $\psi_{x}\in\mathcal{A}_{[x,x+1]}$ for all $x$. Consider the sequence of random Hamiltonians 
\begin{equation}
H_n(\omega)=H_n^0(\omega_0)+\sum_{x=0}^{n-1} \delta_x(\omega_1) \psi_{x}; 
\end{equation}
over the probability space $\Omega_0\times \Omega_1$ equipped with the product measure. If $t_n$ is the transmission time of $H_n$, then for any $\gamma>0$ and $\alpha\in(0,1/3)$ satisfying
\begin{equation}
\eta\left(\frac{1-3\alpha}{1-\alpha}\right)> 2[(\beta+1)\gamma-1] \log\left(\frac{1}{p}\right),
\end{equation}
\begin{equation}
\frac{n^\gamma}{t_n(e^{-\alpha \eta n})}\to 0
\end{equation}
in probability.
\end{theorem}

Unfortunately we do not know how to prove a similar robustness result for models with a decay function $F$ that decays slower than exponentially. For example, certain anisotropic XY chains are only known to exhibit stretched exponential dynamical localization, as we note in Section \ref{subsec_xychain} 

Theorem \ref{thm:mainthm} concerns finite volume Hamiltonians. The following theorem shows that in certain cases one can work directly with the thermodynamic limit.

\begin{theorem}\label{thm:thermolimit}
Suppose that $\Phi_0$ is a random interaction over the probability space $(\Omega_0,\Pr_0)$ whose finite volume Hamiltonians are exponentially dynamically localized. Suppose that $(\delta_x)_{x\in \Ir}$ is a sequence of i.i.d. Bernoulli random variables over the probability space $(\Omega_1,\Pr_1)$, with $\Pr_1(\delta_0=0)=p\in(0,1]$. Let $(\psi_x)_{x\in\Ir}$ denote a uniformly bounded sequence with $\psi_x\in\A_{[x,x+1]}$ for all $x$. Let $\Phi_2$ be the random nearest neighbor interaction given by,
\be
\Phi_2(\{x,x+1\})=\delta_x\psi_x
\ee
for all $x\in\Ir$. Define the random interaction $\Phi(\omega)=\Phi_0(\omega_0)+\Phi_1(\omega_1)$ over the probability space $\Omega_0\times\Omega_1$ equipped with the product measure. If, for almost every $\omega_0\in\Omega_0$, there is a (possibly random) $F$-function $F$ such that $\Phi_0$ is $F$-normed, then the thermodynamic limit, $\tau_t$, of $\Phi$ exists almost surely. For any fixed $r\in \Nl$, define
\be
t_n(\epsilon)=\inf\{|t|: \sup_{\substack{A\in\A^1_{[-r,0]}\\ B\in\A^1_{[n,\infty)}}} \|[\tau_t(A),B]\|>\epsilon\}.
\ee
Then for any $\gamma>0$ and $\alpha\in(0,1/3)$ satisfying
\begin{equation}
\eta\left(\frac{1-3\alpha}{1-\alpha}\right)> 2[(\beta+1)\gamma-1] \log\left(\frac{1}{p}\right),
\end{equation}
\begin{equation}
\frac{n^\gamma}{t_n(e^{-\alpha \eta n})}\to 0
\end{equation}
in probability.
\end{theorem}

\section{Proofs of Main Results}\label{sec:proofs}

\subsection{Proofs of results about LIOMs }\label{sec:proofs_LIOM}

Showing that LIOMs of the fist kind imply dynamical localization is a straightforward application of the quasi-locality properties of the LIOMs.

\begin{proof}[Proof of Theorem \ref{prop:LIOM_LRbound}]
For any $A\in\mathcal{A}_X^1$, $B\in\mathcal{A}_Y^1$,
\begin{equation}
\|[\tau_t^{H_n}(A),B]\|=\|[\tau_t^{\tilde{H_n}}(\tilde{A}),\tilde{B}]\|,
\end{equation}
where $\tilde{O}=U_n^*OU_n$ for an observable $O$. Using the quasi-locality of the unitary $U_n$ specified in Eq. \eqref{liomqloc}, by a standard
application of conditional expectations (see, for example, \cite[Section IV.A]{nachtergaele:2019}),
we can find (random) local observables ${A}_{\lambda}\in\mathcal{A}_{X_{n,\lambda}}$ and ${B}_{\lambda}\in\mathcal{A}_{Y_{n,\lambda}}$, with $\|A_{n,\lambda}\|,\|B_{n,\lambda}\|\leq 1$ such that,
\begin{align}
\|\tilde{A}-A_{\lambda}\|\leq D_{n,X,\lambda}\\
\|\tilde{B}-B_{\lambda}\|\leq D_{n,Y,\lambda},
\end{align}
where $D_{n,X,\lambda}$ and $D_{n,Y,\lambda}$ have the desired expectation bound. Therefore,
\begin{align}\label{eq:liom_LRbound}
\|[\tau_t^{\tilde{H_n}}(\tilde{A}),\tilde{B}]\|&\leq 2\left(D_{X,\lambda,n}+D_{Y,\lambda,n}\right)+\|[\tau_t^{\tilde{H}_n}(A_{\lambda}),B_{\lambda}]\|.
\end{align}
Now,
\begin{equation}
\|[\tau_t^{\tilde{H}_n}(A_{\lambda}),B_{\lambda}]\|=\|[\tau_t^{\tilde{H}_{X,Y}}(A_{\lambda}),B_{\lambda}]\|
\end{equation}
where 
\begin{equation}
\tilde{H}_{X,Y}(\omega)=\sum_{\substack{Z\subset[0,n]:\\ Z\cap X_{n,\lambda},Z\cap Y_{n,\lambda}\neq\varnothing}}
\sum_{{\bf m}\in\prod_{x\in Z} \{2,...,d_x\}}\phi_n({\bf m},Z)\prod_{z\in Z}S_{{\bf m}_z;z}
\end{equation}
Note that $\tilde{H}_{X,Y}$ consist of the terms of $\tilde{H}_n$ which do not in general commute with either $A_{\lambda}$ or $B_{\lambda}$. If $f(t)=[\tau_t^{\tilde{H}_{X,Y}}(A_\lambda),B_\lambda]$, then
\begin{equation}
f'(t)=i[[\tilde{H}_{X,Y},\tau_t^{\tilde{H}_{X,Y}}(A_\lambda)],B_\lambda]=-i[f(t),\tilde{H}_{X,Y}]-i[[B_\lambda,\tilde{H}_{X,Y}],\tau_t^{\tilde{H}_{X,Y}}(A_\lambda)]
\end{equation}
Since the first term on the right is norm preserving, we have that,
\begin{equation}
\|[\tau_t^{\tilde{H}_{X,Y}}(A_\lambda),B_\lambda]\|\leq 4 |t| \|\tilde{H}_{X,Y}\|.
\end{equation}
The estimate,
\begin{align*}
\|\tilde{H}_{X,Y}\| &\leq \sum_{\substack{Z\subset[0,n]:\\ Z\cap X_{n,\lambda},Z\cap Y_{n,\lambda}\neq\varnothing}} \left\|\sum_{{\bf m}\in\prod_{x\in Z} \{2,...,d_x\}}\phi_n({\bf m},Z)\prod_{z\in Z} S_{{\bf m}_z;z} \right\| \\
&\leq \sum_{\substack{x\in X_{n,\lambda}\\y\in Y_{n,\lambda}}}\sum_{\substack{Z:\\x,y\in Z}} \left\|\sum_{{\bf m}\in\prod_{x\in Z} \{2,...,d_x\}}\phi_n({\bf m},Z)\prod_{z\in Z} S_{{\bf m}_z;z} \right\| \leq C_{n}(\omega) \sum_{\substack{x\in X_{n,\lambda}\\y\in Y_{n,\lambda}}} F(|x-y|),
\end{align*}
together with \eqref{eq:liom_LRbound} completes the proof.
\end{proof}

The existence of LIOMs of the second kind for uniform-in-time dynamically localized systems follows from a combination of quasi-locality 
arguments and compactness.

\begin{proof}[Proof of Theorem \ref{thm:thermolioms}]
We first show how to construct a sequence $T_n$ for which the limit in \eq{time_average} exists almost surely for any dynamics that is sufficiently localized uniformly in time.
For $A\in\A^1_X$ and $T>0$, define
$$
A_T = \frac{1}{T}\int_0^T \tau_t(A) dt.
$$
$A_T$ is random since $\tau_t$ is. 

For each $N\in\Nl$, let $\Pi_N$ denote the conditional expectation $\A^{\rm loc} \to \A_{X(N)}$ defined as the limit of the normalized partial trace 
over the complement of $X(N)=\{y\in\Ir: d(y,X)<N\}$ (see \cite[Section 4.2]{nachtergaele:2019}). Since the dynamics $\tau_t$ is assumed to satisfy \eq{uniform_exp_loc}, we
have 
\be 
\E ( \sup_T\Vert \Pi_N(A_T) -A_T\Vert ) \leq C F(N)
\label{quasi-local-time-averages}\ee
where $C=2\chi(|X|)$. In particular, $\sum_{N=1}^\infty F(N)<\infty$ implies that
\be\label{as_convergence}
\lim_N \sup_T \Vert \Pi_N(A_T) -A_T\Vert = 0  \text{ almost surely}
\ee
 Since $\A^1_{X(N)} $ is compact, there exists a sequence $(T^{(N)}_n)_{n\geq 1}$, and $A(N) \in \A^1_{X(N)}$ such that
$$
\lim_n \Pi_N(A_{T^{(N)}_n})= A(N).
$$
We can pick the sequences $(T^{(N)}_n)_{n\geq 1}$ such that $(T^{(N+1)}_n)_{n\geq 1}$ is a subsequence of $(T^{(N)}_n)_{n\geq 1}$, for all $N$.
Fix $\epsilon >0$, and let $N\leq M$. Choose $K(N,M) $ such that for all $n\geq K(N,M)$, we have
$$
\Vert \Pi_N(A_{T^{(N)}_n}) - A(N)\Vert \leq \epsilon , \ \Vert \Pi_M(A_{T^{(M)}_n}) - A(M)\Vert \leq \epsilon.
$$
Since $N\leq M$, $(T^{(M)}_n)_{n\geq 1}$ is a subsequence of $(T^{(N)}_n)_{n\geq 1}$. Therefore, we also have
$$
\Vert \Pi_N(A_{T^{(M)}_n}) - A(N)\Vert \leq \epsilon , \text{ for all } n\geq K(N,M).
$$
Using these bounds we have
\beann
\Vert A(N) - A(M)\Vert &\leq& 2\epsilon + \Vert \Pi_N(A_{T^{(M)}_n}) - \Pi_M(A_{T^{(M)}_n})\Vert\\
&\leq & 2\epsilon + \Vert \Pi_N(A_{T^{(M)}_n}) - A_{T^{(M)}_n}\Vert +  \Vert \Pi_M(A_{T^{(M)}_n}) - A_{T^{(M)}_n}\Vert\\
&\leq& 2\epsilon +  \sup_T \Vert \Pi_N(A_T) -A_T\Vert+ \sup_T \Vert \Pi_M(A_T) -A_T\Vert.
\eeann
Since $\epsilon >0$ is arbitrary, this estimate along with \eq{as_convergence} shows that $(A(N))_N$ is almost surely a Cauchy sequence in $\A_\Ir$. Denote its limit by $\tilde A$.

 We can now pick an increasing sequence $K_N$ such that for all $n \geq K_N$ we have
$$
\Vert  \Pi_N(A_{T^{(N)}_n}) -A(N)\Vert \leq \frac{1}{N}.
$$
Then
$$
\lim_N  \Pi_N(A_{T^{(N)}_{K_N}}) = \lim_N A(N) = \tilde A.
$$
Since we also have
$$
\Vert \Pi_N(A_{T^{(N)}_{K_N}}) -A_{T^{(N)}_{K_N}} \Vert \leq \sup_T \Vert \Pi_N(A_T) -A_T\Vert,
$$
we can conclude the convergence of the sequence of time averages:
\be\label{time-average_convergence} 
\lim_N A_{T^{(N)}_{K_N}} =\tilde A.
\ee

The time-invariance of $\tilde A$ is obvious from the fact that it is the limit of time averages as in \eq{time-average_convergence}.
By taking the $\limsup$ of \eq{quasi-local-time-averages} we also obtain a quasi-locality estimate for $\tilde A$:
\be 
\E ( \Vert [\tilde A, B]\Vert ) \leq C F(d(X,\supp B))
\ee

We can now apply this to $A=h_x$ and, possibly after taking another subsequence, obtain a sequence of times $T_n$ such that for all $x\in\Ir$,
\be
\tilde h_x= \lim_{n\to\infty} \frac{1}{T_n}\int_0^{T_n} \tau_t (h_x) dx.
\ee
are well-defined, time-invariant, and quasi-local. The model is assumed to be finite range, so the constant $C$ can be chosen to be uniform in $x$.

Finally, the quasi-local Hamiltonians $\tilde H_\Lambda$ defined by
$$
\tilde H_\Lambda = \sum_{x\in\Lambda} \tilde h_x,
$$
generate the same dynamics $\tau_t$ in the thermodynamic limit. To see the last point we once more have to argue we can interchange two limits, which we do next.

Let $X$ be finite, $A\in\A_X^1$, and $\epsilon >0$ . Fix a sufficiently large positive integer $M$ such that for all $\Lambda$ containing $X(M)$ we have
\bea
\sum_{x\in \Lambda} [h_x,A]=\delta(A).
\eea
Then, we have
\be\label{difference1}
\Vert \delta(A) - \tilde\delta(A)\Vert \leq \left\Vert \sum_{x\in \Lambda} [h_x , A] -  \sum_{x\in \Lambda} [\tilde h_x , A] \right\Vert 
+ \sum_{x\notin\Lambda}\Vert [\tilde h_x, A ] \Vert
\ee
Then, for any $L,n\in\Nl$, starting from \eq{difference1}, we obtain the following estimate:
\beann
\Vert \delta(A) - \tilde\delta(A)\Vert &\leq&\left\Vert \sum_{x\in X(M+L)} [h_x , A] -  \sum_{x\in X(M+L)} [\tilde h_x , A] \right\Vert 
+\sum_{x\notin X(M+L)}\Vert [ \tilde h_x, A ] \Vert \\
&=& \left\Vert \left[ \left(\sum_{x\in X(M+L)}  \frac{1}{T_n}\int_0^{T_n} \tau^{(X(M+L))}_t(h_x)  \right), A\right] -  \sum_{x\in X(M+L)} [\tilde h_x , A] \right\Vert \\
&&\quad + \sum_{x\notin X(M+L)}\Vert [ \tilde h_x, A ] \Vert \\
&\leq & \sum_{x\in X(M+L)\setminus X(M)}\left( \sup_{t\in\Rl} \Vert [\tau_t^{(X(M+L))}( h_x),A]\Vert + \Vert [\tilde{h}_x,A]\|\right)\\
&&\quad + \left\Vert \sum_{x\in X(M)} \left[ \frac{1}{T_n}\int_0^{T_n} \tau_t^{(X(M+L))}(h_x)\,dt-\tilde{h}_x,A\right]\right \Vert +\sum_{x\notin X(M+L)} \Vert [\tilde{h}_x,A] \Vert
\eeann
Therefore, almost surely
\beann
\Vert\delta(A)-\tilde{\delta}(A)\Vert&\leq& \liminf_{L\to\infty} \sum_{x\notin X(M)} \left( \sup_{t\in\Rl} \Vert [\tau_t^{(X(M+L))}( h_x),A]\Vert + \Vert [\tilde{h}_x,A]\|\right)\\
&&\quad + \left\Vert \sum_{x\in X(M)} \left[ \frac{1}{T_n}\int_0^{T_n} \tau_t(h_x)\,dt-\tilde{h}_x,A\right]\right \Vert
\eeann
Letting $n\to\infty$ in this inequality gives,
$$
\Vert\delta(A)-\tilde{\delta}(A)\Vert\leq  \liminf_{L\to\infty} \sum_{x\notin X(M)} \left( \sup_{t\in\Rl} \Vert [\tau_t^{(X(M+L))}( h_x),A]\Vert + \Vert [\tilde{h}_x,A]\|\right)
$$
almost surely. By Fatou's lemma,
$$
\E  \liminf_{L\to\infty} \sum_{x\notin X(M)} \left( \sup_{t\in\Rl} \Vert [\tau_t^{(X(M+L))}( h_x),A]\Vert + \Vert [\tilde{h}_x,A]\|\right)\leq 4C\sum_{d=M}^\infty F(d)
$$
This upper bound is summable in $M$, therefore,
$$
\lim_{M\to\infty}  \liminf_{L\to\infty} \sum_{x\notin X(M)} \left( \sup_{t\in\Rl} \Vert [\tau_t^{(X(M+L))}( h_x),A]\Vert + \Vert [\tilde{h}_x,A]\|\right)=0
$$
almost surely, which proves that $\delta(A)=\tilde{\delta}(A)$ with probability 1.
\end{proof}

\subsection{Proofs of results about transmission time}

We will prove Theorem \ref{thm:mainthm} by utilizing the interaction picture decomposition of the Heisenberg dynamics $\tau_t^{H_n}=\tau_t^{H^I_n}\circ \tau_t^{H_n^0}$, where $H_n^I$ is the time dependent random Hamiltonian given by,
\begin{equation}\label{inthamiltonian}
H_n^I(\omega,t)=\sum_{x=0}^{n-1} \delta_x(\omega_1)\tau_t^{H_n^0(\omega_0)}(\psi_{x})
\end{equation}
We make use of this decomposition of the dynamics in the following way: for an integer $d_n\in[0,n]$, for any $A\in\mathcal{A}_0^1$ by quasilocality of the the dynamics $\tau_t^{H_n^0}$ we can write
\be\label{eq:quasiloc}
\tau_t^{H_n^0(\omega_0)}(A)=\tilde{A}(\omega_0,t)+E(\omega_0,t),
\ee
where $\supp(\tilde{A})\subset [0,d_n]$, $\|\tilde A\|\leq1$ and 
\be
\| E(\omega_0,t)\|\leq \chi(1)C_{d_n}(\omega_0)(1+|t|^\beta)
\ee
where $\E C_{d_n}\leq e^{-\eta (d_n+1)}$. Eq. \eqref{eq:quasiloc} gives the following bound,
\be\label{eq:splitdynamics}
\sup_{\substack{A\in\mathcal{A}_0^1\\ B\in\mathcal{A}_n^1} }\| [\tau_t^{H_n(\omega)}(A),B]\| \leq 2\chi(1)C_{d_n}(\omega_0)(1+|t|^\beta)+\sup_{\substack{A\in\mathcal{A}_{[0,d_n]}^1\\ B\in\mathcal{A}_n^1} }\|\tau_t^{H_n^I(\omega)}(A),B]\|.
\ee
To proceed we will need to derive a suitable Lieb-Robinson bound for the dynamics $\tau_t^{H_n^I}$. The first step in deriving such a bound is to write $H_n^I$ in terms of a suitable time dependent random interaction.

First we introduce some notation. Let $\Lambda_n=[0,n]$ and $\Lambda_{n;x}(m)=\{y\in\Lambda_n: d(y,\{x,x+1\})\leq m\}$. We write
\be\label{ptracedecomp}
\tau_t^{H_n^0(\omega_0)}(\psi_{x})=\sum_{m\geq 0}\psi_{n;x}^{(m)}(\omega_0,t),
\ee
where 
\be
\psi_{n;x}^{(m)}(t)=\begin{cases}
\text{Tr}_{\mathcal{H}_{\Lambda_n\setminus \Lambda_{n;x}(0)}}\left(\tau_t^{H_n^0}(\psi_{x})\right) & \text{ if }m=0\\
[\text{Tr}_{\mathcal{H}_{\Lambda_n\setminus\Lambda_{n;x}(m)}}-\text{Tr}_{\mathcal{H}_{\Lambda_n\setminus\Lambda_{n;x}(m-1)}}]\left(\tau_t^{H_n^0}(\psi_{x})\right) & \text{ if } m\geq 1
\end{cases}
\ee
Here $\text{Tr}$ denotes the normalized partial trace operator. Note that the sum in Eq. \eqref{ptracedecomp} is actually a finite sum, since $\psi_{n;x}^{(m)}=0$ for any $m$ such that $\Lambda_{n;x}(m-1)=\Lambda_n$.

\begin{proposition}\label{interactiondecay}
$\supp(\psi_{n;x}^{(m)}(t))\subseteq \Lambda_{n;x}(m)$ for all $m\geq0$ and
\begin{equation}
\|\psi_{n;x}^{(m)}(t)\|\leq\begin{cases}
\|\psi_{x}\| & \text{ if }m=0\\
\|\psi_{x}\|C_{n;x}^{(m)} (1+|t|^\beta) & \text{ if } m\geq 1
\end{cases}
\end{equation}
where $C_{n;x}^{(m)}$ is a non-negative random variable satisfying
\begin{equation}
\E C_{n;x}^{(m)}\leq 2\chi(2)e^{-\eta m}
\end{equation}
\end{proposition}

\begin{proof}
$\supp(\psi_{n;x}^{(m)}(t))\subseteq \Lambda_{n;x}(m)$ follows from properties of the partial trace. The bound $\|\psi_{n;x}^{(0)}(t)\|\leq \|\psi_{x}\|$ is immediate. For $m\geq 1$,
\begin{align}
\|\psi_{n;x}^{(m)}(t)\|&\leq \|\tau_t^{H_n^0}(\psi_{x})-\Tr_{\mathcal{H}_{\Lambda_n\setminus\Lambda_{n;x}(m)}}\left(\tau_t^{H_n^0}(\psi_{x})\right)\|\\
&+\|\tau_t^{H_n^0}(\psi_{x})-\Tr_{\mathcal{H}_{\Lambda_n\setminus\Lambda_{n;x}(m-1)}}\left(\tau_t^{H_n^0}(\psi_{x})\right)\|\nonumber\\
&\leq \|\psi_{x}\| \chi(2)\left(C_{\Lambda_n;\Lambda_x,\Lambda_n\setminus \Lambda_{n;x}(m)}+C_{\Lambda_n;\Lambda_x,\Lambda_n\setminus \Lambda_{n;x}(m-1)}\right)|t|^\beta\\
&\equiv \|\psi_{x}\| C_{n;x}^{(m)}(1+ |t|^\beta).
\end{align}
The expectation bound on $C_{n;x}^{(m)}$ follows from the assumptions.
\end{proof}

The decomposition given in Eq. \eqref{ptracedecomp} provides a way to write $H_n^I(t)$ in terms of a random interaction. Define $\Phi_n(\omega,t):\mathcal{P}(\Lambda_n)\to \mathcal{A}_{\Lambda_n}$ by, 
\begin{equation}\label{eq:random_interaction}
\Phi_n(\omega,t)(X)=\sum_{\substack{(x,m):\\ \Lambda_{n;x}(m)=X}}\delta_x(\omega_1) \psi_{n;x}^{(m)}(\omega_0,t).
\end{equation}
Then $H_n^I=\sum_{X\subseteq[0,n]}\Phi_n(X)$ follows from Eq. \eqref{ptracedecomp}. \newline

We will use Theorem 3.1 of \cite{nachtergaele:2019} in order to obtain a Lieb-Robinson bound for the dynamics generated by $H_n^I$. If we apply that theorem directly to $\Phi_n$, with a suitable decaying function $F$, we obtain a Lieb-Robinson bound with a time growth factor of
\be\label{eq:timefactor0}
\exp\left( \int_0^t\sup_{x,y\in[0,n]}\frac{1}{F(|x-y|)}\sum_{\substack{X\subseteq[0,n]\\ x,y\in X}} \|\Phi_n(\omega,s)(X)\|ds\right)
\ee
This will not be of any use to us, as
\be
\sup_{x,y\in[0,n]}\frac{1}{F(|x-y|)}\sum_{\substack{X\subseteq[0,n]\\ x,y\in X}} \|\Phi_n(\omega,s)(X)\|
\ee
will be of order 1 due to the presence of non-zero $\delta_x$. To remedy this we observe that the methods used in \cite{nachtergaele:2019} produce Lieb-Robinson bounds which are independent of on-site terms in the interaction and also do not depend on the dimension of the Hilbert spaces at each site. This allows us to define a new lattice for the model, which is effectively a subset of $[0,n]$, by identifying certain spins which forces certain interaction terms to become on-site terms. As we explain below, we will be able to obtain a better Lieb-Robinson bound using this method.  Specifically, given $\Gamma\subset[0,n]$, we can define the lattice to obtain a Lieb-Robinson bound for the dynamics generated by $H_n^I$ with a time growth factor of
\be\label{eq:timefactor}
\exp\left( \int_0^t\sup_{x,y\in \Gamma}\frac{1}{F(|x-y|)}\sum_{\substack{X\subseteq [0,n]\\x,y\in X}} \|\Phi_n(\omega,s)(X)\|ds\right).
\ee
Note than in Eq. \eqref{eq:timefactor} the supremum in the exponent is taken over pairs of points $x,y\in \Gamma$, as opposed to in Eq. \eqref{eq:timefactor0} where all possible pairs of points in $[0,n]$ enter. The sum in Eq. \eqref{eq:timefactor} therefore excludes any interaction term whose support does not contain a point of $\Gamma$. The arguments for obtaining such a Lieb-Robinson bound given the subset $\Gamma$ are given in detail in the appendix.

It remains to specify how $\Gamma$ should be chosen. We know that intervals $I$ of length $L\sim \log_{1/p}(n)$ with the property that $\delta_x=0$ for all $x\in I$ exist with high probability. The interaction terms $\Phi_n(X)$ decay exponentially in the diameter of $X$, so the sum of all interaction terms linking sites $x,y\in I$ will decay exponentially in the distance $d(\{x,y\},I^c)$. This suggests that we take $\Gamma$ to consist of the intervals $I$ with a collar of length $\ell$ removed from both sides. The interaction terms linking sites $x,y\in \Gamma$ will then decay at least as fast as $e^{-\eta \ell}$. Taking $\ell$ to be a fraction of $L$ leads to power law decay in $n$ of the interaction strength. The following Lemma makes this precise.

\begin{lemma}\label{fnormlemma}
Fix $n\in\mathbb{N}$ and consider the time dependent random interaction $\Phi_n$ given by Eq. \eqref{eq:random_interaction}. Let $\theta\in (0,1)$ be arbitrary. Consider an event $E\subset\Omega_1$ with the following two properties:
\begin{enumerate}
\item[(i)] $(\delta_1,...,\delta_{n-1})$ is fixed on $E$
\item[(ii)]  There are two disjoint intervals $I_j=[a_j,b_j]$, $j=1,2,$ with $|I_j|\geq \theta \log_{1/p}(n)$ such that $\delta_x\big|_E=0$ for each $x\in I_1\cup I_2$.
\end{enumerate}
For $\sigma\in [0,1/2)$, let $\ell=\lfloor \sigma \theta \log_{1/p}(n)\rfloor$ and define the collared intervals $\tilde{I}_j=[a_j+\ell,b_j-\ell]$. Then for any $x,y\in \tilde{I}_1\cup \tilde{I}_2$,
\begin{equation}
\mathbbm{1}_E(\omega_1) \sum_{\substack{X\subseteq[0,n]:\\x,y\in X}}\|\Phi_n(\omega,t)(X)\| \leq B_{E;x,y}(\omega_0)(1+|t|^\beta),
\end{equation}
where there is a constant $\tilde{C}$, depending only on $\eta$, such that $B_{E;x,y}$ satisfies,
\begin{equation}
\E B_{E;x,y} \leq \tilde{C} n^{-\frac{\lambda \eta \sigma \theta }{\log(1/p)}}e^{-(1-\lambda){\eta}\frac{|x-y|}{2}}
\end{equation}
for any $\lambda\in(0,1)$ .
\end{lemma}

\begin{proof} First note that for any points $x<y$ in $[0,n]$ the following inequality holds,
\begin{equation}\label{prefnormbound}
\sum_{\substack{X\subseteq[0,n]:\\x,y\in X}}\|\Phi_n(\omega,t)(X)\|\leq \sum_{z=0}^{n-1}\sum_{\substack{m\geq\\ \max\{|z-x|,|z-y+1|\}}} \delta_z(\omega_1) \|\psi_{n;z}^{(m)}(\omega_0,t)\|.
\end{equation}
This follows from the fact that $\max\{|z-x|,|z-y+1|\}$ is the smallest integer $m$ such that $x,y\in\Lambda_{n;z}(m)$. Without loss of generality assume $a_1<a_2$, and take $x\leq y\in\tilde{I}_2\cup\tilde{I}_2$. Suppose $x\in \tilde{I}_s$, $y\in\tilde{I}_r$ with $s\leq r$. On the event $E$, $\delta_z(\omega_1)=0$ if $z\in I_1\cup I_2$, so we have the bound
\begin{align}
&\mathbbm{1}_E(\omega_1)\sum_{z=0}^{n-1}\sum_{\substack{m\geq\\ \max\{|z-x|,|z-y+1|\}}} \delta_z(\omega_1)\|\psi_{n;z}^{(m)}(\omega_0,t)\|\nonumber\\
&\leq \sum_{z\notin I_1\cup I_2} \sum_{\substack{m\geq\\ \max\{|z-x|,|z-y+1|\}}}\|\psi_{n;z}^{(m)}(\omega_0,t)\|\nonumber\\
&\leq (\sup_{x}\|\psi_{x,x+1}\|)\sum_{z\notin I_r\cup I_s} \sum_{\substack{m\geq\\ \max\{|z-x|,|z-y+1|\}}}C_{n;z}^{(m)}(\omega_0)(1+|t|^\beta)\equiv B_{E;x,y}(\omega_0)(1+|t|^\beta).
\end{align}
By Proposition \ref{interactiondecay},
\begin{align}\label{Bexpectation}
\E B_{E;x,y} \leq 2(\sup_{x}\|\psi_{x}\|)\chi(2)\sum_{z\notin I_r\cup I_s} \sum_{\substack{m\geq\\ \max\{|z-x|,|z-y+1|\}}}e^{-\eta m}
\end{align}
We have,
\begin{align}\label{eq:l32sum}
\sum_{z\notin I_r\cup I_s} \sum_{\substack{m\geq\\ \max\{|z-x|,|z-y+1|\}}}e^{-\eta m}
&=\left(\sum_{z=0}^{a_r-1}+\sum_{z=b_r+1}^{a_s-1}+\sum_{z=b_s+1}^{n-1}\right)\sum_{\substack{m\geq\\ \max\{|z-x|,|z-y+1|\}}}e^{-\eta m}
\end{align}
We first estimate,
\begin{align}
\left(\sum_{z=0}^{a_r-1}+\sum_{z=b_s+1}^{n-1}\right)\sum_{\substack{m\geq\\ \max\{|z-x|,|z-y+1|\}}}e^{-\eta m}&\leq \left[ \sum_{z=0}^{a_r-1} \sum_{m= y-z-1}^\infty e^{-\eta m}+\sum_{z=b_s+1}^{n-1}\sum_{m= z-x}^\infty e^{-\eta m}\right]\nonumber\\
&\leq \sum_{k=y-a_r}^\infty \sum_{m=k}^\infty e^{-\eta m}+\sum_{k=b_s-x}^\infty\sum_{m=k}^\infty e^{-\eta m}\nonumber\\
&= \frac{1}{(1-e^{-\eta})^2} [e^{-{\eta}(y-a_r)}+e^{-{\eta}(b_s-x)}]\nonumber \\
&\leq  \frac{1}{(1-e^{-\eta})^2}e^{-{\eta}[(y-x)+\ell]},
\end{align}
where we used that $b_s-y,x-a_r\geq \ell$ in the last line. The remaining sum in Eq. \eqref{eq:l32sum} vanishes when $r=s$. If $r<s$ then,
\begin{align}
\sum_{z=b_r+1}^{a_s-1} \sum_{\substack{m\geq\\ \max\{z-x,y-z-1\}}}e^{-\eta m}&\leq \sum_{z=b_r+1}^{\lceil \frac{y+x-1}{2}\rceil -1}\sum_{m=y-z-1}^\infty e^{-\eta m}+\sum_{z=\lceil\frac{y+x-1}{2}\rceil}^{a_s-1}\sum_{m=z-x}^\infty e^{-\eta m}\nonumber\\
&\leq \sum_{k=y-\lceil \frac{y+x-1}{2}\rceil }^\infty\sum_{m=k}^\infty e^{-\eta m}+ \sum_{k=\lceil \frac{y+x-1}{2}\rceil -x}^\infty \sum_{m=k}^\infty e^{-\eta m}\nonumber\\
&\leq \frac{1}{(1-e^{-\eta})^2}[e^{-{\eta}(y-\lceil \frac{y+x-1}{2}\rceil)}+e^{-{\eta}(\lceil \frac{y+x-1}{2}\rceil-x)}]\nonumber\\
&\leq \frac{1}{(1-e^{-\eta})^2}e^{\frac{\eta}{2}} e^{-{\eta}(\frac{y-x}{2})}
\end{align}
If $r<s$, then $|x-y|\geq 2\ell$ and
\begin{align}
e^{-{\eta} (\frac{y-x}{2})}=e^{-{\eta}\lambda ( \frac{y-x}{2})}e^{-{\eta}(1-\lambda)(\frac{y-x}{2})}\leq e^{-\lambda{\eta}\ell}e^{-(1-\lambda){\eta}(\frac{y-x}{2})}
\end{align}
Therefore,
\begin{align}
\sum_{z\notin I_1 \cup I_2} \sum_{\substack{m\geq\\ \max\{|z-x|,|z-y+1|\}}}e^{-\eta m}\leq \frac{1}{(1-e^{-\eta})^2}\left(e^{-\eta(|x-y|+\ell)}+(1-\delta_{s,r})e^{\frac{\eta}{2}}e^{-\lambda\eta\ell}e^{-(1-\lambda)\eta\frac{|x-y|}{2}}\right),
\end{align}
which together with Eq. \eqref{Bexpectation} proves the lemma.
\end{proof}

We now use Lemma \ref{fnormlemma} to prove that a Lieb-Robinson bound holds for the dynamics $\tau_t^{H_n^I}$ on an event contained in $E$ which has probability nearly that of $E$ for large $n$.

\begin{lemma}\label{lemma:intLRbound}
Assume the hypotheses and notation of Lemma \ref{fnormlemma}, with the additional assumption that $|I_j|\leq \frac{3}{2}\log_{1/p}(n)$ for $j=1,2$. Then for any $\nu\in(0,1)$ there is an event $W_E\subset\Omega_0$ such that for any $\xi\in(0,1)$ there are positive constants $c_0$ and $c_1$, which depend only on $\nu,\xi,\lambda$ and $\eta$, such that

\begin{equation}\label{eq:conditionalLRbound}
\mathbbm{1}_{W_E}(\omega_0)\mathbbm{1}_E(\omega_1)\sup_{\substack{A\in\mathcal{A}_{[0,a_1]}^1\\ B\in\mathcal{A}_n^1}}\|[\tau_t^{H^I_n(\omega)}(A),B]\|\leq c_0(e^{c_1 n^{-\frac{\nu\lambda \eta \sigma \theta }{\log(1/p)}}(|t|+|t|^{\beta+1})}-1)e^{-\xi\frac{\nu(1-\lambda)\eta}{2}d(I_1,I_2)}.
\end{equation}
Furthermore, the event $W_E$ satisfies,
\begin{align}
\Pr(W_E)\geq 1- \tilde{C}' n^{-\frac{(1-\nu) \lambda \eta \sigma \theta}{\log(1/p)}} \log_{1/p}(n)
\end{align}
where
\be
\tilde{C}'=\frac{3\tilde{C}}{1-e^{-\frac{(1-\nu)(1-\lambda)\eta}{2}}}
\ee
\end{lemma}

\begin{proof}
For a fixed pair $x,y$ in $\tilde{I}_1\cup \tilde{I}_2$, by Markov's inequality and Lemma \ref{fnormlemma},
\begin{equation}
\Pr(B_{E;x,y}\leq n^{-\frac{\nu \lambda \eta \sigma \theta}{\log(1/p)}}e^{-\nu(1-\lambda)\eta \frac{|y-x|}{2}})\geq 1-\tilde{C} n^{-\frac{(1-\nu) \lambda \eta \sigma \theta}{\log(1/p)}}e^{-(1-\nu)(1-\lambda)\eta \frac{|x-y|}{2}}
\end{equation}
Let 
\begin{equation}
W_E=\{ B_{E;x,y}\leq n^{-\frac{\nu \lambda \eta \sigma \theta}{\log(1/p)}}e^{-\nu(1-\lambda)\eta \frac{|y-x|}{2}} \text{ for all }x,y\in \tilde{I}_1\cup\tilde{I}_2\}
\end{equation}
It follows that,
\begin{align}
\Pr( W_E)\nonumber&\geq 1-\tilde{C}n^{-\frac{(1-\nu) \lambda \eta \sigma \theta}{\log(1/p)}}\sum_{\substack{x\leq y:\\ x,y\in\tilde{I}_1\cup \tilde{I}_2}}e^{-(1-\nu)(1-\lambda)\eta \frac{|x-y|}{2}}\nonumber\\
&\geq 1- \tilde{C}n^{-\frac{(1-\nu) \lambda \eta \sigma \theta}{\log(1/p)}} \sum_{x\in \tilde{I}_1\cup \tilde{I}_2}\sum_{y=x}^\infty e^{-(1-\nu)(1-\lambda)\eta \frac{|x-y|}{2}}\nonumber\\
&\geq 1-\frac{3\tilde{C}}{1-e^{-\frac{(1-\nu)(1-\lambda)\eta}{2}}} n^{-\frac{(1-\nu) \lambda \eta \sigma \theta}{\log(1/p)}} \log_{1/p}(n)
\end{align}
Let $F$ be any $F$-function on $\Ir$ such that for any $c>0$,
\begin{equation}
\sup_{x\in \Ir} \frac{e^{-c|x|}}{F(|x|)}<\infty.
\end{equation}
Then by Lemma \ref{fnormlemma} and the definition of $W_E$ we have that,
\begin{align}
&\mathbbm{1}_{E\cap W_E}(\omega) \sup_{x,y\in \tilde{I}_1\cup \tilde{I}_2} \frac{1}{e^{-\xi \frac{\nu (1-\lambda) \eta}{2}|x-y|}F(|x-y|)}\sum_{\substack{X\subseteq[0,n]:\\x,y\in X}}\|\Phi_n(\omega,t)(X)\|\nonumber\\
&\leq n^{-\frac{\nu \lambda \eta \sigma \theta}{\log(1/p)}}\sup_{x,y\in \tilde{I}_1\cup\tilde{I}_2} \frac{e^{-\nu(1-\xi)(1-\lambda)\eta \frac{|y-x|}{2}}}{F(|x-y|)}\nonumber\\
&\leq n^{-\frac{\nu\lambda\eta \sigma\theta}{\log(1/p)}}\sup_{x\in \Ir} \frac{e^{-\frac{\nu(1-\xi)(1-\lambda)\eta}{2}|x|}}{F(|x|)} 
\end{align}
The result now follows from Proposition \ref{LRbound} in the appendix, using the collection $\mathcal{I}=\{\tilde{I}_1,\tilde{I}_2\}$.
\end{proof}

From Lemma \ref{lemma:intLRbound}, we see that the best Lieb-Robinson bound will be obtained on events $E$ where the intervals $I_1$ and $I_2$ are as far apart as possible. This in fact occurs with high probability: Let $\theta\in(0,1)$ and suppose $F_n$ is the event that there are two intervals of consecutive 0's of length at least $\theta\log_{1/p}(n)$ in $n$ i.i.d. Bernoulli trials, such that the distance $r_n$ between the intervals satisfies $\lim r_n/n=1$. Then the probability of $F_n$ tends to 1 as $n$ tends to infinity. This can be seen by noting that if $\theta'\in (\theta,1)$, then the longest run $R_n$ of zeros in $\lfloor n^{\theta'}\rfloor $ i.i.d. Bernoulli trials has the property that
\be
\frac{R_n}{\theta' \log_{1/p}(n)}\to 1
\ee
in probability. Therefore, with a probability tending to 1, there is an interval of length at least $\theta\log_{1/p}(n)$ in both the first and last $\lfloor n^{\theta'}\rfloor$ trials in $n$ Bernoulli trials. The distance between these two intervals is at least $n-2n^{\theta'}$.\newline

\begin{proof}[Proof of Theorem \ref{thm:mainthm}]
We will prove the result for $\beta>0$. The case $\beta=0$ requires only minor modifications. We will show that under the hypotheses of the theorem there is a sequence of events $Q_n$ with $\lim_{n\to \infty}\Pr(Q_n)=1$, and a deterministic sequence $x_n$ satisfying $\lim_{n\to\infty}n^\gamma/x_n=0$ such that
\be
\mathbbm{1}_{Q_n} t_n(e^{-\alpha \eta n})\geq x_n.
\ee
From this it easily follows that $n^\gamma/t_n(e^{-\alpha \eta n})\to0$ in probability. \newline

Let $\kappa\in (\alpha,1)$. Our starting point is Eq. \eqref{eq:splitdynamics}, with $d_n=\lfloor \kappa n \rfloor$. Consider the event $F_n=\{ 2\chi(1)C_{d_n}\leq n^{-(\gamma\beta+1)} e^{-\alpha \eta n}\}$. By Markov's inequality,
\be\label{eq:locdynbound}
\Pr(F_n)\geq 1-2\chi(1)n^{\gamma\beta+1}e^{-(\kappa-\alpha)\eta n}.
\ee
It follows from Eq. \eqref{eq:splitdynamics} that,
\begin{equation}\label{eq:splitLRbound}
\mathbbm{1}_{F_n}(\omega_0)\sup_{\substack{A\in\mathcal{A}_0^1\\ B\in\mathcal{A}_n^1} }\| [\tau_t^{H_n(\omega)}(A),B]\| \leq (1+|t|^\beta)n^{-\gamma\beta+1}e^{-\alpha\eta n}+\sup_{\substack{A\in\mathcal{A}_{[0,d_n]}^1\\ B\in\mathcal{A}_n^1} }\|\tau_t^{H_n^I(\omega)}(A),B]\|.
\end{equation}

Choose $\theta\in (0,1)$, and let $G_n\subset \Omega_1$ denote a sequence of events in which there are two runs of zeros in the list $(\delta_{d_n},...,\delta_{n-1})$ of length at least $\theta \log_{1/p}(n)$ and no more than $\tfrac{3}{2}\log_{1/p}(n)$, and such that if $r_n$ denotes the distance between the two runs, $\lim_{n\to\infty}r_n/n\to (1-\kappa)$. We have observed that such a sequence can be chosen with $\lim_{n\to\infty}\Pr(G_n)=1$. Write,
\begin{equation}
G_n=\bigsqcup_{E\in\mathcal{F}_n}E,
\end{equation}
where $\mathcal{F}_n$ is the set of events $E\subset\Omega_1$ on which $(\delta_{d_n},\delta_{d_n+1},...\delta_{n-1})$ is fixed. Consider an event $E\in\mathcal{F}_n$. By Lemma \ref{lemma:intLRbound} we have that,
\begin{equation}\label{eq:rintLRbound}
\mathbbm{1}_{W_E}(\omega_0)\mathbbm{1}_E(\omega_1)\sup_{\substack{A\in\mathcal{A}_{[0,d_n]}^1\\ B\in\mathcal{A}_n^1}}\|[\tau_t^{H^I_n(\omega)}(A),B]\|\leq c_0(e^{c_1 n^{-\frac{\nu\lambda \eta \sigma \theta }{\log(1/p)}}(|t|+|t|^{\beta+1})}-1)e^{-\xi \frac{\nu (1-\lambda)\eta}{2}r_n}.
\end{equation}

Note that Eq. \eqref{eq:locdynbound} and Lemma \ref{lemma:intLRbound} imply that for each $E\in \mathcal{F}_n$,
\be\label{eq:problowerbound}
\Pr(W_E \cap F_n)\geq 1-2\chi(1)n^{\gamma\beta+1}e^{-(\kappa-\alpha)\eta n)}-\tilde{C}'n^{-\frac{(1-\nu) \lambda \eta \sigma \theta}{\log(1/p)}} \log_{1/p}(n)\equiv X_n.
\ee
Clearly $X_n\to 1$ as $n\to\infty$. Now define $Q_n=\sqcup_{E\in\mathcal{F}_n} E\cap W_E\cap F_n$. By independence and Eq. \eqref{eq:problowerbound},
\be\label{eq:highprob}
\Pr(Q_n)=\sum_{E\in\mathcal{F}_n}\Pr(E)\Pr(W_E\cap F_n)\geq X_n \sum_{E\in\mathcal{F}_n}\Pr(E)=X_n\Pr(G_n),
\ee
which shows that $\Pr(Q_n)\to 1$ as $n\to\infty$. \newline

We now show that the transmission time has a deterministic lower bound on the event $Q_n$. Eqs. \eqref{eq:splitLRbound} and  \eqref{eq:rintLRbound} give the bound,
\begin{multline}\label{eq:highprobLRbound}
\mathbbm{1}_{Q_n}(\omega)\sup_{\substack{A\in\mathcal{A}_0^1\\ B\in\mathcal{A}_n^1} }\| [\tau_t^{H_n(\omega)}(A),B]\| \leq (1+|t|^\beta)n^{-(\gamma\beta+1)}e^{-\alpha\eta n}\\
+c_0(e^{c_1 n^{-\frac{\nu\lambda \eta \sigma \theta }{\log(1/p)}}(|t|+|t|^{\beta+1})}-1)e^{-\xi \frac{\nu (1-\lambda)\eta}{2}r_n}
\end{multline}
It follows that
\be
\mathbbm{1}_{Q_n}t_n(e^{-\alpha \eta n})\geq \min\{ (\tfrac{1}{2}n^{\gamma \beta+1}-1)^{\frac{1}{\beta}}, Y_n\}\equiv x_n,
\ee
where
\be
Y_n=\left[ \frac{n^{\frac{\lambda\eta\sigma\theta}{\log(1/p)}}}{2 c_1}\log \left( 1+\frac{1}{2c_0}e^{(\xi \frac{\nu(1-\lambda)}{2}\frac{r_n}{n}-\alpha)\eta n} \right) \right]^{\frac{1}{\beta+1}}.
\ee
Since $\lim_{n\to\infty}r_n/n= (1-\kappa)$, we have that
\be\label{eq:yassymp}
\frac{n^{(\frac{\nu\lambda\eta\sigma\theta}{\log(1/p)}+1)/(\beta+1)}}{Y_n}
\ee
converges to a positive constant, provided 
\be\label{eq:constraint1}
\xi\frac{\nu(1-\lambda)(1-\kappa)}{2}>\alpha
\ee
One can check that Eq. \eqref{eq:constraint1} can be satisfied only if $\alpha<1/3$. In this case, $\nu$ and $\xi$ close to 1 can be chosen so Eq. \eqref{eq:constraint1} is satisfied only if
\be\label{eq:constraint2}
\kappa\in (\alpha,1-2\alpha) \text{ and } \lambda\in (0, 1-\frac{2\alpha}{1-\kappa}).
\ee
If Eq. \eqref{eq:constraint1} is satisfied, then Eq. \eqref{eq:yassymp} implies that
\be
\lim_{n\to\infty}\frac{n^\gamma}{Y_n}=0
\ee
provided 
\be\label{eq:constraint3}
\eta>\frac{[\gamma(\beta+1)-1]}{\nu\lambda \sigma\theta}\log(1/p).
\ee
We conclude that if Eq. \eqref{eq:constraint3} is satisfied, then $n^\gamma/x_n\to 0$.\newline

We can choose parameters so Eq. \eqref{eq:constraint3} is satisfied if $\eta$ is larger than
\be
\inf \frac{[\gamma(\beta+1)-1]}{\nu\lambda \sigma\theta}\log(1/p)= \frac{2[\gamma(\beta+1)-1]}{1-\frac{2\alpha}{1-\alpha}}\log(1/p),
\ee
where the infimum is taken over parameter values satisfying Eqs. \eqref{eq:constraint2} and \eqref{eq:constraint1}.
\end{proof}
The following general proposition is needed to adapt the proof of Theorem \ref{thm:mainthm} to the thermodynamic limit.

\begin{proposition}\label{prop:intpic}
Let $\Phi_1,\Phi_2:\mathcal{P}_0(\Ir)\to\A_\Ir^\text{loc}$ are two $F$-norm bounded interactions with respect to some $F$-function. Let $H^j_\Lambda=\sum_{X\subseteq\Lambda}\Phi_j(X)$ denote the corresponding local Hamiltonians for each finite volume $\Lambda\subset\Ir$. Let $\tau_t$ denote the thermodynamic limit of the model $\Phi_1+\Phi_2$. Then the following limit holds,
\be
\tau_t=\lim_{\Lambda_2\uparrow\Ir}\lim_{\Lambda_1\uparrow\Ir} \tau_t^{H_{\Lambda_1}^1+H_{\Lambda_2}^2}
\ee
where the limits are taken along any increasing, exhaustive sequences of finite subsets of $\Ir$. For each finite $\Lambda\subset\Ir$, $\lim_{\Lambda_1\uparrow\Ir}\tau_t^{H^1_{\Lambda_1}+H_\Lambda^2}$ can be expressed in terms of the interaction picture:
\be\label{eq:intdecomp}
\lim_{\Lambda_1\uparrow\Ir} \tau_t^{H_{\Lambda_1}^1+H_\Lambda^2}=\tau_t^{\Lambda,I}\circ\tau_t^0,
\ee
where $\tau_t^0$ is the thermodynamic limit of the model $\Phi_1$, and $\tau_t^{\Lambda,I}$ is the dynamics generated by the time-dependent, quasi-local Hamiltonian $\tau_t^0(H_\Lambda^2)$.
\end{proposition}

Armed with Proposition \ref{prop:intpic}, the proof of Theorem \ref{thm:thermolimit} is nearly identical to the proof of Theorem \ref{thm:mainthm}. Using the decomposition \eqref{eq:intdecomp}, one can show that the bound \eqref{eq:conditionalLRbound} in Lemma \ref{lemma:intLRbound} holds with $\tau_t^{H_n^I}$ replaced by $\tau_t^{\Lambda,I}$, uniformly for intervals $\Lambda\supseteq[0,n]$. One can then obtain the bound \eqref{eq:highprobLRbound} with $\tau_t^{H_n}$ replaced by $\tau_t^{I;\Lambda}\circ \tau_t^0$. Taking the limit $\Lambda\uparrow\Ir$ gives this bound for the thermodynamic limit, and the proof proceeds exactly as before.

\section{Applications}\label{sec:applications}

As mentioned before, MBL in the sense of dynamical localization without an energy restriction, has been rigorously established only for the random $XY$ chain and partial results exists for the quantum Ising chain. Naturally, applications of the results in this paper, at the moment, are also restricted to these two models. An extension we will not discuss in detail here is to fermion chains. Our arguments go through without change as long the same obvious analogous conditions are satisfied.
Generalizing in another direction, one could  consider non-random quasi-periodic chains with localization properties such as the Fibonacci chain \cite{mace:2019} or the fermion models studied by Mastropietro \cite{mastropietro:2017,mastropietro:2017a}.

\subsection{The Disordered XY Chain}\label{subsec_xychain}

Consider three real-valued sequences $\mu_j,\gamma_j$ and $\omega_j$. These sequences may be random. The finite volume anisotropic XY Hamiltonian in an external field in the $z$-direction is given by the Hamiltonian
\begin{equation}
H_n^{XY}=\sum_{j=0}^{n-1}\mu_j[(1+\gamma_j)\sigma_j^x\sigma_{j+1}^x+(1-\gamma_j)\sigma_j^y\sigma_{j+1}^y]+\lambda\sum_{j=0}^{n}\omega_j \sigma_j^z,
\end{equation}
acting on $\bigotimes_{x=0}^n\Cx^2$. Here $\sigma_j^x,\sigma_j^y,\sigma_j^z\in\mathcal{A}_{j}$ denote the Pauli spin matrices acting on the $j$th spin. It is well known that the many-body XY Hamiltonian can be written in terms of an effective one-body Hamiltonian via the Jordan-Wigner transformation \cite{lieb:1961}:
\be
H_n^{XY}=\mathcal{C}^*M_n\mathcal{C},
\ee
where $\mathcal{C}^t=(c_0,...,c_n,c_0^*,...,c_n^*)$ is a column vector of operators $c_j$ given by
$$
c_j=\frac{1}{2}(\sigma_j^x-i\sigma_j^y)\prod_{k=0}^{j-1}\sigma_k^{z}, 
$$
and $M_n$ is a 2$\times$2 block matrix,
$$
M_n=\left(\begin{array}{cc}A_n & B_n \\-B_n & -A_n\end{array}\right)
$$
with
$$
A_n=\left(\begin{array}{ccccc}\omega_0 & -\mu_0 & 0 & 0 & 0 \\-\mu_0 & \ddots & \ddots & 0 & 0 \\0 & \ddots & \ddots & \ddots & 0 \\0 & 0 & \ddots & \ddots & -\mu_n \\0 & 0 & 0 & -\mu_n & \omega_n\end{array}\right),
$$
and
$$
B_n=\left(\begin{array}{ccccc}0 & -\mu_0\gamma_0 & 0 & 0 & 0 \\\mu_0\gamma_0 & \ddots & \ddots & 0 & 0 \\0 & \ddots & \ddots & \ddots & 0 \\0 & 0 & \ddots & \ddots & -\mu_n\gamma_n \\0 & 0 & 0 & \mu_n\gamma_n & 0\end{array}\right).
$$
The following result was proved in \cite{hamza:2012}:
\begin{theorem}\label{hamzathm}
Suppose that the matrices $M_n$ are exponentially dynamically localized in the following sense: there exist positive constants $C$ and $\eta$ such that for any integers $n\geq 0$ and $j,k\in[0,n+1]$,
\be\label{xyeff}
\E [\sup_{t\in\Rl}  | (e^{-it M_n})_{j,k}|+|(e^{-it M_n})_{j,n+k+1}|]\leq Ce^{-\eta |j-k|}.
\ee
Then the Heisenberg dynamics $\tau_t^{H_n^{XY}}$ of the $XY$-chain is exponentially dynamically localized, uniformly in time, with $\chi(x)=4^x$.
\end{theorem}

Theorem \ref{hamzathm} shows that if the sequences $\mu_j,\gamma_j$ and $\omega_j$ are such that dynamical localization for the $M_n$ holds, then Theorem \ref{thm:mainthm} applies to the XY chain. If, in addition $\sup_j \mu_j$ and $\sup_j\gamma_j$ are almost surely finite, then the XY chain satisfies the hypotheses of Theorem \ref{thm:thermolimit}.

There are several instances in which the matrices $M_n$ are known to satisfy \eq{xyeff}. For example, if $\gamma_j=0$ and $\mu_j=1$ for all $j$, and the $\omega_j$ are i.i.d. with compactly supported density, then $B_n=0$ and $A_n$ is the finite volume Anderson model. In this case it is well known that \eq{xyeff} holds \cite{kunz:1980}. In \cite{elgart:2014} a large class of random block operators were shown to exhibit exponential dynamical localization at high disorder. Under the assumption that $\mu_j$ and $\gamma_j$ are deterministic and bounded, and that the $\omega_j$ are i.i.d. with sufficiently smooth distribution, this class of random block operators includes $M_n$ and \eq{xyeff} holds for sufficiently large $|\lambda|$. Therefore in these models the conditions of theorems \ref{thm:thermolioms}, \ref{thm:mainthm} and \ref{thm:thermolimit} are satisfied.

The anisotropic case was also investigated in \cite{chapman:2015}. The methods there prove localization of the $M_n$ for $\omega_j$ with compactly supported distribution contained in $(-\infty,-2)$ or $(2,\infty)$. For these results smoothness of the distribution is not needed, however the method produces a bound with a stretched exponential, not an exponential as in \eq{xyeff}. This localization bound is shown to imply a uniform in time localization bound for the XY chain where the decay is given by a stretched exponential. Therefore disordered anisotropic XY models have LIOMs, as shown by Theorem \ref{thm:thermolioms}, but our results do not imply robustness of long transmission times under perturbation.

\subsection{The Quantum Ising Chain}

Another model that has been widely discussed in the literature is the quantum Ising with random coefficients. For concreteness, consider the following family of Hamiltonians for a spin-$1/2$ systems on a chain $[a,b]\subset \Ir$:
\be
H_{[a,b]} = \sum_{x=a}^{b-1} J_x \sigma_x^3\sigma_{x+1}^3 + \sum_{x=a}^b \gamma\Gamma_x \sigma_x^1 + h_x\sigma_x^3,
\label{QIsingHamiltonian}\ee
where $(J_x), (\Gamma_x),$ and $(h_x)$ are three independent sequences of i.i.d. random variables, each with bounded density of compact support.

Mathematical work by John Imbrie and a variety of numerical results point towards the existence of a description of this model in terms of LIOMs of the first kind (Definition \ref{def:LIOM}). To state the various claims we need to introduce the assumptions made by Imbrie \cite{imbrie:2016}. Let $\lambda_\alpha^{[a,b]}$ denote an enumeration of the eigenvalues, which are almost surely simple.

\noindent
{\em Imbrie's Assumption:} There exist $\gamma_0$, such that for all $\gamma \in (-\gamma_0,\gamma_0)$, there exists
constants $\nu , C>0$, such that for all $\delta>0,a <b \in \Ir$ we have
\be
\Pr (\min_{\alpha\neq \beta} |\lambda_\alpha^{[a,b]} - \lambda_\beta^{[a,b]}|<\delta)\leq \delta^\nu C^{b-a+1}.
\label{ImbriesAssumption}\ee

In \cite{imbrie:2016} Imbrie uses a systematic perturbation theory which, under his assumptions, he argues combines with a multi-scale 
analysis to prove detailed properties about the eigenvectors of the Hamiltonians $H_{[a,b]}$ for sufficiently small $\gamma$, uniformly in the length of the chain. We should note, however, that among experts in the multiscale analysis approach to proving localization there is no agreement that such an argument can indeed be carried out along the lines described in \cite{imbrie:2016}.

In the review paper \cite[Section 4.3]{imbrie:2017} the following implications of the perturbation analysis of \cite{imbrie:2016} are stated: $H_{[a,b]}$ is 
diagonalized by a quasi-local unitary transformation and the resulting energy eigenvalues when labeled by Ising configurations take 
the form of a random Ising model with multi-spin interactions of strong decay, i.e., something very similar to the LIOM picture we define in Definition \ref{def:LIOM}. The LIOM representation is explained by starting from Imbrie's localization property for the eigenvectors $\psi^{[a,b]}$ which reads as follows: there exists $\kappa>0$ such that for all sufficiently long finite intervals $[a,b]$ containing the origin one has
$$
\left| 1- \E \left[ \sum_\alpha \rho_\alpha |\langle\psi_\alpha, \sigma^3_0\psi_\alpha\rangle|\right] \right|\leq \gamma^\kappa,
$$
where $\rho_\alpha$ is a probability distribution such as
$$
\rho_\alpha = \frac{e^{-\beta \lambda^{[a,b]}_\alpha}}{\sum_\gamma  e^{-\beta \lambda^{[a,b]}_\gamma}}.
$$

In the spirit of these results it appears that the disordered quantum Ising chain may indeed be a model where the exponential dynamical localization of Definition \ref{def:dynlocaliz} and the LIOM picture of Definition \ref{def:LIOM} indeed both hold.

\appendix

\section{Lieb-Robinson Bounds}
In this appendix we develop a bound on the velocity of propagation under the Heisenberg dynamics which ignores interaction terms supported in a given subset of the lattice. We use the results of \cite{nachtergaele:2019}, in which Lieb-Robinson bounds which do not depend on on-site interactions are developed for Hamiltonians expressed in terms of time-dependent interactions. 

Let $(\Gamma,d)$ denote a countable metric space, and let $\mathcal{P}_0(\Gamma)$ denote the collection of finite subsets of $\Gamma$. Assign a spin Hilbert space $\mathcal{H}_x$ to each $x\in \Gamma$. The algebra of local observables is given by $\mathcal{A}^\text{loc}=\cup_{X\in \mathcal{P}_0(\Gamma)} \mathcal{A}_X$, where $\mathcal{A}_X=\bigotimes_{x\in X}\mathcal{B}(\mathcal{H})$. A time-dependent interaction $\Phi:\Rl\times \mathcal{P}_0(\Gamma)$ is called continuous if $t\mapsto \Phi(t,X)$ is norm continuous for every $X\in\mathcal{P}_0(\Gamma)$.

To measure the spatial decay of the interaction we introduce the notion of an $F$-function. Let $(\Gamma,d)$ denote a countable metric space. Then an $F$-function on $(\Gamma,d)$ is a function $F:[0,\infty)\to(0,\infty)$ such that
\begin{enumerate}
\item $F$ is non-increasing.
\item $F$ is integrable, i.e.,
\begin{equation}
\|F\|=\sup_{x\in\Gamma}\sum_{y\in\Gamma} F(d(x,y))<\infty.
\end{equation}
\item $F$ satisfies the convolution identity,
\begin{equation}
C_F=\sup_{x,y\in\Gamma}\frac{1}{F(d(x,y))}\sum_{z\in\Gamma}F(d(x,z))F(d(z,y))<\infty.
\end{equation}
\end{enumerate}
If $\mu>0$, it is easy to show that $F_\mu(x)=e^{-\mu x}F(x)$ also defines an $F$ function on $(\Gamma,d)$ with $\|F_\mu\|\leq \|F\|$ and $C_{F_\mu}\leq C_F$.\newline

Given an $F$-function $F$, we denote by $\mathcal{B}_F$ the set of continuous interactions $\Phi:\Rl\times\mathcal{P}_0(\Gamma)\to \mathcal{A}^\text{loc}$ such that the function on $\Rl$
\be
t\mapsto \sup_{x,y\in \Gamma} \frac{1}{F(d(x,y))} \sum_{\substack{x,y\in X\\ |X|>1}}\|\Phi(t,X)\|
\ee
is locally bounded.

\begin{theorem}[Theorem 3.1 in \cite{nachtergaele:2019}]\label{thm:monsterLRbound}
Let $\Phi\in\mathcal{B}_{F_\mu}$ for some $F$-function $F$ and $\mu>0$, and let $X,Y\in\mathcal{P}_0(\Gamma)$ with $X\cap Y=\varnothing$. Then for any $\Lambda\in\mathcal{P}_0(\Gamma)$ with $X\cup Y\subseteq \Lambda$, we have
\be
\sup_{\substack{A\in\mathcal{A}_X^1\\ B\in\mathcal{A}_Y^1}}\| [\tau_t^{H_\Lambda}(A),B]\| \leq \frac{2 \|F\|} {C_{F_\mu}}\min\{|X|,|Y|\} (e^{2 C_{F_\mu} I(t)}-1)e^{-\mu d(X,Y)}
\ee 
for every $t\in \Rl$, where
\be
I(t)=\int_{\min\{0,t\}}^{\max\{0,t\}} \sup_{x,y\in \Gamma} \frac{e^{\mu d(x,y)}}{F(d(x,y))} \sum_{\substack{x,y\in X\\ |X|>1}}\|\Phi(s,X)\|ds.
\ee
\end{theorem}

We will now apply the previous theorem to obtain a Lieb-Robinson bound which ignores interaction terms in certain parts of the lattice. For simplicity we restrict ourselves to one-dimensional finite volume systems. Neither of these restrictions is essential.

Suppose that we have a quantum spin chain $\mathcal{H}=\bigotimes_{x=0}^n\mathcal{H}_x$ on the interval $\Lambda_n=[0,n]\subset\Ir_+$ together with a time-dependent Hamiltonian $H(t)$ generated by an interaction $\Phi(t):\mathcal{P}(\Lambda_n)\to\mathcal{B}(\mathcal{H})$. Let $\mathcal{I}=\{I_j\}_{j=1}^m$ be a collection of disjoint subintervals $I_j=[a_j,b_j]\subset \Lambda_n$, satisfying $b_{j}<a_{j+1}$. For purposes of notation let $b_0=0$ and $a_{m+1}=n$. We seek to define an equivalent spin chain in which the spins located on the sites $[b_j,a_{j+1}]$ are identified. Define the contracted lattice $\Gamma_{\mathcal{I}}$ by,
\[
\Gamma_{\mathcal{I}}=\cup_{j=1}^m [a_j,b_j)\cup\{n\}
\]
Define a map $\mathcal{C}:\Lambda_n\to \Gamma_{\mathcal{I}}$ by,
\begin{equation}
\mathcal{C}(x)=\begin{cases}
a_{j} & \text{ if }x\in[b_{j-1},a_{j}] \text{ for some }j=1,2,...,m+1\\
x & \text{ Otherwise}
\end{cases}
\end{equation}
Note that $\mathcal{C}$ maps a site in $\Lambda_n$ to its corresponding site in $\Gamma_\mathcal{I}$. For each $x\in\Gamma_{\mathcal{I}}$, define
\begin{equation}
\mathcal{H}_x'=
\bigotimes_{z\in\mathcal{C}^{-1}(\{x\})}\mathcal{H}_z 
\end{equation}
Then $\bigotimes_{x=0}^n\mathcal{H}_x=\bigotimes_{x\in\Gamma_{\mathcal{I}}}\mathcal{H}_x'$, and an observable which has support $X$ in $\mathcal{A}_{\Lambda_n}$ has support $\mathcal{C}(X)$ in $\mathcal{A}_{\Gamma_\mathcal{I}}$. Define an interaction $\tilde\Phi(t)$ on $\Gamma_{\mathcal{I}}$ by,
\begin{equation}
\tilde\Phi(t)(X)=\sum_{\substack{Z\subseteq\Lambda_n\\ \mathcal{C}(Z)=X}}\Phi(t)(Z)
\end{equation}
Then $\tilde{\Phi}$ and $\Phi$ generate the same Hamiltonian. With this setup we have the following proposition.

\begin{theorem}\label{LRbound}
Suppose $d$ is a metric on $\Gamma_{\mathcal{I}}$. Let $\mu>0$ and let $F$ denote any $F$-function on $(\Gamma_{\mathcal{I}},d)$. Then for any $X,Y\subseteq\Lambda_n$ with $\mathcal{C}(X)\cap\mathcal{C}(Y)=\varnothing$ we have,
\begin{equation}
\sup_{\substack{A\in\mathcal{A}_{X}^1\\ B\in\mathcal{A}_Y^1}}\|[\tau_t^{H}(A),B]\|\leq \frac{2\|F\|}{C_{F_\mu}}\min\{|\mathcal{C}(X)|,|\mathcal{C}(Y)|\}(e^{2C_{F_\mu} I(t)}-1)e^{-\mu d(\mathcal{C}(X),\mathcal{C}(Y))}
\end{equation}
holds for all $t\in\Rl$, where
\begin{equation}\label{contractedfnorm}
I(t)=\int_{\min\{0,t\}}^{\max\{0,t\}}\sup_{x,y\in\Gamma_\mathcal{I}}\frac{e^{\mu d(x,y)}}{F(d(x,y))}\sum_{\substack{X\subseteq\Gamma_{\mathcal{I}}:\\x,y\in X, \\ |X|>1}}\|\tilde\Phi(s)(X)\|ds.
\end{equation}
\end{theorem}

\begin{proof}
Apply Theorem \ref{thm:monsterLRbound} to the spin model $\tilde{\Phi}$. 
\end{proof}

A few remarks about this theorem need to be made. Note that
\begin{equation}
\sum_{\substack{X\subseteq\Gamma_\mathcal{I}:\\x,y\in X,\\ |X|>1}}\|\tilde{\Phi}(t)(X)\| =\sum_{\substack{X\subseteq\Gamma_\mathcal{I}:\\x,y\in X,\\ |X|>1}} \| \sum_{\substack{Z\subseteq\Lambda_n\\ \mathcal{C}(Z)=X}}\Phi(t)(Z) \|\leq \sum_{\substack{Z\subseteq\Lambda_n:\\ x,y\in Z,\\ |\mathcal{C}(Z)|>1}}\|\Phi(t)(Z)\|
\end{equation}
for any pair $x,y\in\Gamma_\mathcal{I}$. If $Z\subset[b_{j-1},a_j]$ for some $j$, then $\mathcal{C}(Z)$ will contain at most one point of $\Gamma_\mathcal{I}$. Therefore Theorem \ref{LRbound} provides an upper bound on the speed of propagation which excludes elements from the original interaction with support $Z$.

While Theorem \ref{LRbound} was stated for an arbitrary metric $d$ on $\Gamma_\mathcal{I}$, there are two natural metrics which both allow $(\Gamma_\mathcal{I},d)$ to be isometrically embedded into $\Ir_+$. One choice to simply restrict the usual metric on $\Ir_+$ to $\Gamma_\mathcal{I}$. Another choice is to define $d$ so that $(\Gamma_{\mathcal{I}},d)$ isometrically embeds into $[0,L]$, where $L=\sum_{j=1}^m(b_j-a_j)$. With either of these metrics, given an $F$-function $F$ on $\Ir_+$ with the usual metric, the constants in Theorem \ref{LRbound} can be chosen to be $c_0=2\|F\|/C_{F_\mu}$ and $c_1=2 C_{F_\mu}$. In particular, these constants do not depend on $n$ or the collection of intervals $\mathcal{I}$. This follows from the fact that $\Gamma_\mathcal{I}$ isometrically embeds into $(\Ir_+,|\cdot|)$ when equipped with either of these metrics.

\section*{Acknowledgments}

This article reports on work supported by the National Science Foundation under Grants DMS-1207995, DMS-1515850 and DMS-1813149. We also acknowledge support from the Centre de Recherches Math\'ematiques (Montr\'eal) and the Simons Foundation during Fall 2018, when part of this work was carried out. Our work was stimulated by fruitful discussions with Gunter Stolz and Simone Warzel. 

\providecommand{\bysame}{\leavevmode\hbox to3em{\hrulefill}\thinspace}
\providecommand{\MR}{\relax\ifhmode\unskip\space\fi MR }
\providecommand{\MRhref}[2]{%
  \href{http://www.ams.org/mathscinet-getitem?mr=#1}{#2}
}
\providecommand{\href}[2]{#2}


\begin{thebibliography}{10}

\bibitem{abdul-rahman:2017}
H~Abdul-Rahman, B.~Nachtergaele, R.~Sims, and G.~Stolz, \emph{Localization
  properties of the disordered xy spin chain. a review of mathematical results
  with an eye toward many-body localization}, Ann. Phys. (Berlin) \textbf{529}
  (2017), 1600280, arXiv:1610.01939.

\bibitem{aizenman:1993}
M.~Aizenman and S.~Molchanov, \emph{Localization at large disorder and at
  extreme energies: An elementary derivation}, Commun. Math. Phys. \textbf{157}
  (1993), 245--278.

\bibitem{aizenman:2009}
M.~Aizenman and S.~Warzel, \emph{Localization bounds for multiparticle
  systems}, Commun. Math. Phys. \textbf{290} (2009), 903--934.

\bibitem{aizenman:2015}
\bysame, \emph{Random operators. disorder effects on quantum spectra and
  dynamics.}, Graduate Studies in Mathematics, vol. 168, Amer. Math. Soc.,
  2015.

\bibitem{anderson:1958}
P.~W. Anderson, \emph{Absence of diffusion in certain random lattices}, Phys.
  Rev. \textbf{109} (1958), 1492--1505.

\bibitem{basko:2006}
D.~M. Basko, I.~L. Aleiner, , and B.~L. Altshuler, \emph{Metal-insulator
  transition in a weakly interacting many-electron system with localized
  single-particle states}, Annals of Physics \textbf{321} (2006), 1126--1205.

\bibitem{beaud:2017}
V.~Beaud and S.~Warzel, \emph{Low-energy {F}ock-space localization for
  attractive hard-core particles in disorder}, Ann. H. Poincar\'{e} \textbf{18}
  (2017), 3143--3166.

\bibitem{beaud:2018}
\bysame, \emph{Bounds on the entanglement entropy of droplet states in the
  {XXZ} spin chain}, J. Math. Phys. \textbf{59} (2018), 012109,
  arXiv:1709.10428.

\bibitem{braun:2019}
P.~Braun, D.~Waltner, M.~Akila, B.~Gutkin, and T.~Guhr, \emph{Transition from
  quantum chaos to localization in spin chains}, .arXiv:1902.06265, 2019.

\bibitem{chandran:2015}
A.~Chandran, I.H Kim, G.~Vidal, and D.A Abanin, \emph{Constructing local
  integrals of motion in the many-body localized phase}, Phys. Rev. B
  \textbf{91} (2015), 085425.

\bibitem{chapman:2015}
J.~Chapman and G.~Stolz, \emph{Localization for random block operators related
  to the {XY} spin chain}, Ann. H. Poincar\'{e} \textbf{16} (2015), 405--435.

\bibitem{chen:2019}
C.-F. ~Chen, and A.~Lucas,
\emph{Finite Speed of Quantum Scrambling with Long Range Interactions},
Phys. Rev. Lett. \textbf{123} (2019), 250605.

\bibitem{chulaevsky:2009}
V.~Chulaevsky and Y.~Suhov, \emph{Multi-particle {A}nderson localisation:
  induction on the number of particles}, Math. Phys. Anal. Geom. \textbf{12}
  (2009), 117--139.

\bibitem{de-roeck:2017a}
W.~De~Roeck and F.~Huveneers, \emph{Stability and instability towards
  delocalization in many-body localization systems}, Phys. Rev. B \textbf{95}
  (2017), 155129.

\bibitem{de-roeck:2016}
W.~De~Roeck, F.~Huveneers, M.~M\"uller, and M.~Schiulaz, \emph{Absence of
  many-body mobility edges}, Phys. Rev. B \textbf{93} (2016), 014203.

\bibitem{de-roeck:2020}
W.~De~Roeck, F.~Huveneers, and S.~Olla, 
\emph{Subdiffusion in One-Dimensional {H}amiltonian Chains with Sparse Interactions},
J. Stat. Phys. (2020), online, DOI 10.1007/s10955-020-02496-1

\bibitem{de-roeck:2017}
W.~De~Roeck and J.Z. Imbrie, \emph{Many-body localization: stability and
  instability}, Phil. Trans. R. Soc. A \textbf{375} (2017), 20160422.

\bibitem{de-roeck:2015}
W.~De~Roeck and M.~Sch\"utz, \emph{Local perturbations perturb
  exponentially--locally}, J. Math. Phys. \textbf{56} (2015), 061901.

\bibitem{elgart:2018b}
A.~Elgart, A.~Klein, and G.~Stolz, \emph{Droplet localization in the random
  {XXZ} model and its manifestations}, J. Phys. A: Math. Theor. \textbf{51}
  (2018), 01LT02.

\bibitem{elgart:2018a}
\bysame, \emph{Manifestations of dynamical localization in the disordered xxz
  spin chain}, Commun. Math. Phys. \textbf{361} (2018), 1083--1113 |,
  arXiv:1708.00474.

\bibitem{elgart:2018}
\bysame, \emph{Many-body localization in the droplet spectrum of the random
  {XXZ} quantum spin chain}, J. Funct. Anal. \textbf{275} (2018), 211--258,
  arXiv:1703.07483.

\bibitem{elgart:2014}
A.~Elgart, M.~Shamis, and S.~Sodin, \emph{Localisation for non-monotone
  {S}chr\"{o}dinger operators}, J. Eur. Math. Soc. \textbf{16} (2014),
  909--924.

\bibitem{germinet:2001}
F.~Germinet and A.~Klein, \emph{Bootstrap multiscale analysis and localization
  in random media}, Commun. Math. Phys. \textbf{222} (2001), 415--448.

\bibitem{goihl:2019}
M.~Goihl, J.~Eisert, and C.~Krumnow, \emph{Are many-body localized systems
  stable in the presence of a small bath?}, arXiv:1902.0437, 2019.

\bibitem{hamza:2012}
E.~Hamza, R.~Sims, and G.~Stolz, \emph{Dynamical localization in disordered
  quantum spin systems}, Commun. Math. Phys. \textbf{315} (2012), 215--239.

\bibitem{huse:2014}
D.A. Huse, R.~Nandkishore, and V.~Oganesyan, \emph{Phenomenology of fully
  many-body-localized systems}, Phys. Rev. B \textbf{90} (2014), 174202.

\bibitem{imbrie:2016}
J.Z. Imbrie, \emph{On many-body localization for quantum spin chains}, J. Stat.
  Phys. \textbf{163} (2016), 998--1048.

\bibitem{imbrie:2017}
J.Z. Imbrie, V.~Ros, and A.~Scardicchio, \emph{Review: Local integrals of
  motion in many-body localized systems}, Ann. Phys. (Berlin) \textbf{529}
  (2017), 1600278.

\bibitem{kunz:1980}
H.~Kunz and B.~Souillard, \emph{Sur le spectre des op\'erateurs aux
  diff\'erences finies al\'eatoires}, Commun. Math. Phys. \textbf{78} (1980),
  201--224.

\bibitem{lieb:1961}
E.~Lieb, T.~Schultz, and D.~Mattis, \emph{Two soluble models of an
  antiferromagnetic chain}, Ann.~Phys.~(N.Y.) \textbf{16} (1961), 407--466.

\bibitem{luitz:2017}
D.J. Luitz, F.~Huveneers, and W.~De~Roeck, \emph{How a small quantum bath can
  thermalize long localized chains}, Phys. Rev. Lett. \textbf{119} (2017),
  150602.

\bibitem{mace:2019}
N.~Mac\'{e}, N.~Laflorencie, and F.~Alet, \emph{Many-body localization in a
  quasiperiodic fibonacci chain}, SciPost Phys. \textbf{6} (2019), 050.


\bibitem{mastropietro:2017a}
V.~Mastropietro, \emph{Coupled identical localized fermionic chains with
  quasi-random disorder}, Phys. Rev. B \textbf{95} (2017), 075155.

\bibitem{mastropietro:2017}
\bysame, \emph{Localization in interacting fermionic chains with quasi-random
  disorder}, Commun. Math. Phys. \textbf{351} (2017), 283--309.

\bibitem{nachtergaele:2006}
B.~Nachtergaele, Y.~Ogata, and R.~Sims, \emph{Propagation of correlations in
  quantum lattice systems}, J. Stat. Phys. \textbf{124} (2006), 1--13.

\bibitem{nachtergaele:2019}
B.~Nachtergaele, R.~Sims, and A.~Young, \emph{Quasi-locality bounds for quantum
  lattice systems. {I}. {L}ieb-{R}obinson bounds, quasi-local maps, and
  spectral flow automorphisms}, J. Math. Phys. \textbf{60} (2019), 061101.

\bibitem{oganesyan:2007}
V.~Oganesyan and D.A. Huse, \emph{Localization of interacting fermions at high
  temperature}, Phys. Rev. B \textbf{75} (2007), 155111.

\bibitem{pal:2010}
A.~Pal and D.A Huse, \emph{The many-body localization phase transition}, Phys.
  Rev. B \textbf{82} (2010), 174411.

\bibitem{serbyn:2013a}
M.~Serbyn, Z.~Papi\'{c}, and D.~A. Abanin, \emph{Local conservation laws and
  the structure of the many-body localized states}, Phys. Rev. Lett.
  \textbf{111} (2013), 127201.

\bibitem{serbyn:2013}
M.~Serbyn, Z.~Papic, and D.~A. Abanin, \emph{Universal slow growth of
  entanglement in interacting strongly disordered systems}, Phys. Rev. Lett.
  \textbf{110} (2013), 260601.

\bibitem{sims:2016}
R.~Sims and S.~Warzel, \emph{Decay of determinantal and {P}faffian correlation
  functionals in one-dimensional lattices}, Commun. Math. Phys. \textbf{347}
  (2016), 903--931.

\bibitem{suntajs:2019}
J.~S\v{u}ntajs, J.~Bon\v{c}a, T.~Prosen, and L.~Vidmar, \emph{Quantum chaos
  challenges many-body localization}, arXiv:1905.06345, 2019.

\bibitem{thiery:2018}
T.~Thiery, F.~Huveneers, M.~M\"uller, and W.~De~Roeck, \emph{Many-body
  delocalization as a quantum avalanche}, Phys. Rev. Lett. \textbf{121} (2018),
  140601.

\end{thebibliography}
\end{document}